\documentclass[journal]{IEEEtran}

\usepackage[latin1]{inputenc}
\usepackage{amsmath}
\usepackage{amssymb}
\usepackage{graphicx}

\usepackage{enumitem}
\usepackage[font=footnotesize,caption=false]{subfig}
\usepackage{stfloats}
\usepackage{mathrsfs}
\usepackage{cite}
\usepackage{color}
\usepackage{tikz}
\usepackage{multirow}

\newtheorem{theorem}{Theorem}

\newtheorem{definition}{Definition}
\newtheorem{lemma}{Lemma}
\newtheorem{example}{Example}

\makeatletter
\long\def\@makecaption#1#2{\ifx\@captype\@IEEEtablestring%
\footnotesize\begin{center}{\normalfont\footnotesize #1}\\
{\normalfont\footnotesize\scshape #2}\end{center}%
\@IEEEtablecaptionsepspace
\else
\@IEEEfigurecaptionsepspace
\setbox\@tempboxa\hbox{\normalfont\footnotesize {#1.}~~ #2}%
\ifdim \wd\@tempboxa >\hsize%
\setbox\@tempboxa\hbox{\normalfont\footnotesize {#1.}~~ }%
\parbox[t]{\hsize}{\normalfont\footnotesize \noindent\unhbox\@tempboxa#2}%
\else
\hbox to\hsize{\normalfont\footnotesize\hfil\box\@tempboxa\hfil}\fi\fi}
\makeatother

\mathchardef\ordinarycolon\mathcode`\:
  \mathcode`\:=\string"8000
  \begingroup \catcode`\:=\active
    \gdef:{\mathrel{\mathop\ordinarycolon}}
  \endgroup

\IEEEoverridecommandlockouts

\begin{document}

\title{Low-Density Parity-Check Codes From Transversal Designs With Improved Stopping Set Distributions}

\author{Alexander~Gruner,~\IEEEmembership{Student Member,~IEEE}, and Michael~Huber,~\IEEEmembership{Member,~IEEE} %
\thanks{Manuscript received June 26, 2012; revised December 4, 2012, and February 23, 2013.
The work of A.~Gruner was supported within the Promotionsverbund `Kombinatorische Strukturen' (LGFG Baden-W\"urttemberg).
The work of M.~Huber was supported by the Deutsche Forschungsgemeinschaft (DFG) via a Heisenberg grant (Hu954/4) and a Heinz Maier-Leibnitz Prize grant (Hu954/5).}
\thanks{The authors are with the Wilhelm Schickard Institute for Computer Science, Eberhard Karls Universit\"at T\"ubingen, Sand~13,
D-72076 T\"ubingen, Germany (corresponding author's e-mail address: michael.huber@uni-tuebingen.de).}}%

\maketitle

\begin{abstract}
This paper examines the construction of low-density parity-check (LDPC) codes from transversal designs based on sets of mutually orthogonal Latin squares (MOLS). By transferring the concept of configurations in combinatorial designs to the level of Latin squares, we thoroughly investigate the occurrence and avoidance of stopping sets for the arising codes. Stopping sets are known to determine the decoding performance over the binary erasure channel and should be avoided for small sizes. Based on large sets of simple-structured MOLS, we derive powerful constraints for the choice of suitable subsets, leading to improved stopping set distributions for the corresponding codes. We focus on LDPC codes with column weight~4, but the results are also applicable for the construction of codes with higher column weights. Finally, we show that a subclass of the presented codes has quasi-cyclic structure which allows low-complexity encoding.
\end{abstract}


\section{Introduction}

\IEEEPARstart{A}{lthough} a big effort has been made to construct LDPC codes achieving the limit of Shannon's coding theorem, it is still a challenging and contemporary task to develop practical codes with an improved decoding performance. For this, there are two usual methods. First, by increasing the girth of the code's Tanner graph (e.g., \cite{LuMoura06, high_girth_10}), more precisely, by avoiding small cycles, which are known to be harmful for the iterative decoding process. Large girth speeds the convergence of iterative decoding and leads to better performance if the number of iterations is limited. Second, by lowering the so-called \emph{error-floors}. This phenomenon is a significant flattening of the \emph{bit-error-rate} (BER) curve beyond a certain \emph{signal-to-noise-ratio} (SNR). Lower error-floors can be achieved by the avoidance of specific substructures that act as internal states in which the decoder can be trapped. Depending on the channel, these substructures have quite different characteristics. Over the \emph{binary erasure channel} (BEC), such states are known as stopping sets. These combinatorial entities completely determine the decoding performance over the BEC \cite{Di02}.

An ongoing line of research for designing LDPC codes focuses on the construction of codes with structured parity-check matrices, allowing the use of low-complexity encoding algorithms as opposed to random-like code constructions. Moreover, these codes can guarantee structural properties like 4-cycle-free Tanner graphs and are more amenable for extensive mathematical analysis. A fertile and sophisticated approach is to apply the well-known concepts of combinatorial design theory to the construction of LDPC codes. Examples are codes based on certain balanced incomplete block designs (BIBDs) called \emph{Steiner 2-designs} \cite{crc_handbook}, an approach that has been intensively studied in the literature \mbox{(e.g., \cite{Vasic04, Vasic01, JohnWell2001, resolvable_designs_for_ldpc_codes, GrunHub12})}. 
A wide range of structured LDPC codes can also be derived from the field of \emph{partial geometries} (e.g., \cite{JohnWell2004}), including the codes from Steiner 2-designs as a subclass. Further subclasses are codes on \emph{generalized quadrangles} as presented in \cite{Vontobel2001}, and codes on \emph{transversal designs} considered in \cite{JohnWell2004,JohnsonDiss} (cf. Subsection~\ref{related} for further discussion). 

This paper is concerned with an exhaustive investigation of stopping sets in LDPC codes from transversal designs. 
Although constructions of such codes were considered for arbitrary column weights, previous examinations primarily focus on column weight 3 codes from transversal designs without \emph{Pasch-configurations}, i.e., avoiding codewords of weight 4.
By contrast, this paper focuses on codes with column weight 4 and higher. 
The transversal designs arise from \emph{mutually orthogonal Latin squares} (MOLS) which provide a simple setting for investigating substructures such as stopping sets and simplify the derivation of codes that avoid some of these substructures.
Here, the MOLS possess a simple structure and are completely determined by some scale factors. Depending on the choice of these scale factors, they result in quite different stopping set distributions. The primary aim of this paper is therefore to identify well-matching choices of scale factors to reveal optimal stopping set distributions, leading to codes with significantly lower error-floors. 

The paper is organized as follows: In Section~\ref{preliminaries}, we give some standard material from combinatorial design theory and coding theory that are important for our purposes. In Section~\ref{link_stopsets_latinSquares}, we transfer the principle of configurations (as a generalization of stopping sets) to the level of Latin squares, leading to the concept of (partial) subrectangles and in the case of stopping sets to full-correlating subrectangles.
In Section~\ref{codes_on_simple_structured_MOLS}, we introduce a class of MOLS with simple structure and investigate the occurrences of full-correlating subrectangles, leading to some bounds for the stopping distance of the arising codes. 
Based on this class of MOLS, we develop a novel code design in Section~\ref{improved_codes}, producing LDPC codes with improved stopping set distributions. Section~\ref{encoding} is concerned with the low-complexity encoding of the arising codes. In Section~\ref{simulations}, we verify the improvement of these codes by extensive simulations. Finally, we discuss the relationship of the presented codes with related code classes in Section~\ref{related}, and Section~\ref{conclusion} concludes the paper.

\section{Preliminaries}\label{preliminaries}

\subsection{Latin Squares}

A \emph{Latin square} $L$ of order $n$ is an array of $n\times n$ cells, where each row and each column contains every symbol of an $n$-set $S$ exactly once \cite{crc_handbook}. Let $L[x,y]$ denote the symbol at row $x\in X$ and column $y\in Y$, where $X$ and $Y$ are $n$-sets indexing the rows and columns of $L$, respectively. Two Latin squares $L_1$ and $L_2$ of order $n$ are \emph{orthogonal}, if they share a common row and column set $X$ and $Y$, respectively, and if the ordered pairs $(L_1[x,y],L_2[x,y])$ are unique for all $(x,y) \in X \times Y$. A set of Latin squares $L_1,...,L_m$ is called \emph{mutually orthogonal}, if for every $1\leq i<j\leq m$, $L_i$ and $L_j$ are orthogonal. These are also referred to as \emph{MOLS}, \emph{mutually orthogonal Latin squares}.

\subsection{Transversal Designs}

A \emph{transversal design} TD$(k,n)$ of order (or group size) $n$ and block size $k$ is a triple $(\mathcal{P},\mathcal{G},\mathcal{B})$, where

\begin{enumerate}[label=(\arabic*),topsep=0pt,parsep=0pt]
	\item $\mathcal{P}$ is a set of $kn$ points.
	\item $\mathcal{G}$ is a partition of $\mathcal{P}$ into $k$ classes of size $n$, called \emph{groups}.
	\item $\mathcal{B}$ is a collection of $k$-subsets of $\mathcal{P}$, called \emph{blocks}.
	\item Every unordered pair of points from $\mathcal{P}$ is contained either in exactly one group or in exactly one block (cf.  \cite{crc_handbook}).
\end{enumerate}

It follows from (1)-(4) that any point of $\mathcal{P}$ occurs in exactly $n$ blocks and that the number of blocks must be $|\mathcal{B}|=n^2$.

\begin{theorem}\label{equiv_MOLS_TDs}
For $k\ge 3$, the existence of a set of $m:=k-2$ mutually orthogonal Latin squares (MOLS) of order $n$ is equivalent to the existence of a TD$(k,n)$ \cite{crc_handbook, BosShri90, Wilson74}.
\end{theorem}

\begin{proof}
We will outline the proof of this known result, since it is important for the understanding of our paper. 
Let $L_1,\hdots,L_m$ be $m$ MOLS with symbol sets $S_1,\hdots,S_m$, and with common row and column sets $X$ and $Y$, respectively. We may assume that the sets $X,Y,S_1,\hdots,S_m$ are pairwise disjoint, which can easily be achieved by renaming the elements. Then we obtain a TD$(k,n)$ with points $\mathcal{P}=\{X \cup Y \cup S_1 \cup \hdots \cup S_m\}$, groups $\mathcal{G}=\{X,Y,S_1,\hdots,S_m\}$ and blocks $\mathcal{B}=\{\{x, y, L_1[x,y], \hdots, L_{m}[x,y]\}: (x,y)\in X\times Y\}$. This process can be reversed to recover a set of $m=k-2$ MOLS from a TD$(k,n)$ for $k\geq 3$.
\end{proof}

Every TD$(k,n)$ can be described by a binary $|\mathcal{P}|\times|\mathcal{B}|$ \emph{incidence matrix} $\mathcal{N}$, with rows indexed by the points of $\mathcal{P}$, columns indexed by the blocks of $\mathcal{B}$, and
\[\mathcal{N}_{i,j}=\begin{cases}
  1,  & \text{if the } i^{\text{th}} \text{ point is in the } j^{\text{th}} \text{ block}\\
  0, & \text{otherwise}.
\end{cases}\]

\begin{example}
Fig.~\ref{fig:MOLS} depicts the incidence matrix of the transversal design TD$(4,5)$ which is equivalent to the orthogonal Latin squares given in the same figure by using the correspondence detailed in the proof of Theorem~\ref{equiv_MOLS_TDs}.
\end{example}

\begin{figure}[t!]
\renewcommand{\arraystretch}{1.2}
\subfloat{
		\scalebox{0.97}{$
		\begin{array}[b]{|ccccc|}
		\hline
		0&1&2&3&4\\
		1&2&3&4&0\\
		2&3&4&0&1\\
		3&4&0&1&2\\
		4&0&1&2&3\\
		\hline
		\end{array}
		$}
}
\subfloat{
		\scalebox{0.97}{$
		\begin{array}[b]{|ccccc|}
		\hline
		0&1&2&3&4\\
		2&3&4&0&1\\
		4&0&1&2&3\\
		1&2&3&4&0\\
		3&4&0&1&2\\
		\hline
		\end{array}
		$}
}
\subfloat{
	\includegraphics[scale = 0.54]{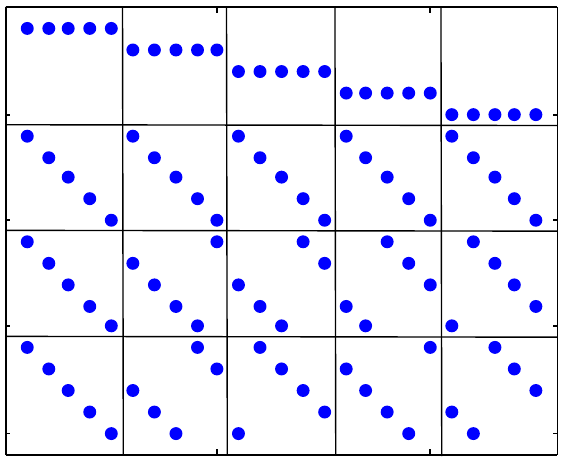}
}
\caption{Orthogonal Latin squares of order 5 and the resulting TD$(4,5)$}
\label{fig:MOLS}
\end{figure}

\subsection{Configurations in Transversal Designs}\label{configuration_in_designs}

Let TD$(k,n)$ be a transversal design with point set $V$ and block set $\mathcal{B}$. By a $(q,l)$-\emph{configuration}, we mean a subset of $l$ blocks of $\mathcal{B}$ whose union contains precisely $q$ points of $V$. With respect to the geometric approach, the blocks are also called \emph{lines}. The \emph{degree} of a point is the number of lines containing the point. If every point has degree at least $2$, the configuration is \emph{full}. 

Note that the concept of configurations can be equally transferred to any combinatorial design, for example in \cite{Col09}, where configurations are introduced for Steiner triple systems.

\subsection{LDPC Codes from Transversal Designs}\label{codes_from_TDs}

The incidence matrix of a TD$(k,n)$ with points $\mathcal{P}$ and blocks $\mathcal{B}$ can directly be used as the parity-check matrix $H$ of a \textit{TD LDPC code}, such that the $|\mathcal{P}|=kn$ points correspond to the parity-check equations of $H$ and the $|\mathcal{B}|=n^2$ blocks correspond to the code bits. The resulting TD LDPC code has block length $N=n^2$, rate $R\geq(n-k)/n$, and a parity-check matrix $H$ with column weight $k$ and row weight $n$. The column weight $k$ corresponds to the block size of TD$(k,n)$ and the row weight $n$ arises from the fact that every point is incident to exactly $n$ blocks. The associated Tanner graph of $H$ is free of 4-cycles and has girth $g=6$ (e.g.~\cite{JohnWell2004}).
Note that these codes were first considered in \cite{JohnWell2004} as a subclass of codes from partial geometries.

\begin{figure*}[t!]
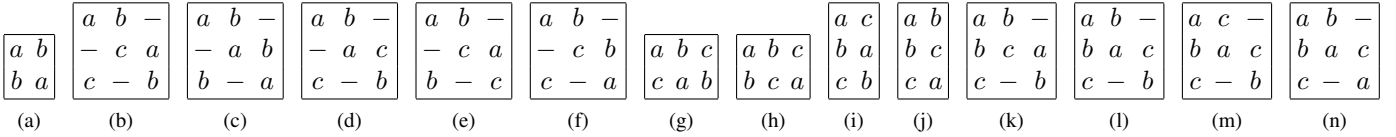

\setlength{\arraycolsep}{2.15pt}
\centerline{
\subfloat[]{
	\begin{array}[b]{|cc|}
	\hline
	a & b \\
	b & a \\
	\hline
	\end{array}
}
\subfloat[]{
	\begin{array}[b]{|ccc|}
	\hline
	a & b & -\\
	- & c & a\\
	c & - & b\\
	\hline
	\end{array}
}
\subfloat[]{
	\begin{array}[b]{|ccc|}
	\hline
	a & b & -\\
	- & a & b\\
	b & - & a\\
	\hline
	\end{array}
}
\subfloat[]{
	\begin{array}[b]{|ccc|}
	\hline
	a & b & -\\
	- & a & c\\
	c & - & b\\
	\hline
	\end{array}
}
\subfloat[]{
	\begin{array}[b]{|ccc|}
	\hline
	a & b & -\\
	- & c & a\\
	b & - & c\\
	\hline
	\end{array}
}
\subfloat[]{
	\begin{array}[b]{|ccc|}
	\hline
	a & b & -\\
	- & c & b\\
	c & - & a\\
	\hline
	\end{array}
}
\subfloat[]{
	\begin{array}[b]{|ccc|}
	\hline
	a & b & c\\
	c & a & b\\
	\hline
	\end{array}
}
\subfloat[]{
	\begin{array}[b]{|ccc|}
	\hline
	a & b & c\\
	b & c & a\\
	\hline
	\end{array}
}
\subfloat[]{
	\begin{array}[b]{|cc|}
	\hline
	a & c\\
	b & a\\
	c & b\\
	\hline
	\end{array}
}
\subfloat[]{
	\begin{array}[b]{|cc|}
	\hline
	a & b\\
	b & c\\
	c & a\\
	\hline
	\end{array}
}
\subfloat[]{
	\begin{array}[b]{|ccc|}
	\hline
	a & b & -\\
	b & c & a\\
	c & - & b\\
	\hline
	\end{array}
}
\subfloat[]{
	\begin{array}[b]{|ccc|}
	\hline
	a & b & -\\
	b & a & c\\
	c & - & b\\
	\hline
	\end{array}
}
\subfloat[]{
	\begin{array}[b]{|ccc|}
	\hline
	a & c & -\\
	b & a & c\\
	c & - & b\\
	\hline
	\end{array}
}
\subfloat[]{
	\begin{array}[b]{|ccc|}
	\hline
	a & b & -\\
	b & a & c\\
	c & - & a\\
	\hline
	\end{array}
}
}
\caption{Full subrectangles up to $7$ cells with symbols $a,b,c$. The subrectangles (a)-(f) are polygons (Theorem~\ref{polygon_condition}).}
\label{upto7cellconfigs}
\end{figure*}

\subsection{Stopping Sets in TD LDPC Codes}

A transversal design can be described by an incidence matrix, which is, in our case, also the parity-check matrix of a TD LDPC code. By intersecting the rows and columns associated with the points and blocks of a full $(q,l)$-configuration (Subsection~\ref{configuration_in_designs}), we obtain a submatrix with $q$ rows, $l$ columns and row weights of at least 2. 
Furthermore, this $(q\times l)$-submatrix can be unambiguously associated with a subgraph of the code's Tanner graph, which is an equivalent representation of the code's parity-check matrix. In this context, the resulting subgraph with $l$ variable nodes and $q$ check nodes is referred to as a \emph{stopping set} of size $l$.

The well-known concept of stopping sets was first introduced in \cite{Di02} and was studied in a various number of subsequent papers (e.g., \cite{Ka03}, \cite{SchwVar06}). 
From a coding theoretic perspective, stopping sets are mostly defined as special subgraphs of the code's Tanner graph; precisely, a stopping set is a subset $\Sigma$ of the variable nodes such that every neighboring check node of $\Sigma$ is connected to $\Sigma$ at least twice (e.g. \cite{SchwVar06}). This definition of stopping sets clearly coincides with the design theoretic approach of full configurations in transversal designs.
The \emph{stopping distance} $s_{min}$ is the size of the smallest stopping set occurring in the code's parity-check matrix $H$. Note that the stopping distance depends on the specific choice of the parity-check matrix for a given code and not on the code itself \cite{SchwVar06}. By the \textit{stopping set distribution} we mean the multiplicities of the occurring stopping sets grouped by their sizes.

\begin{figure}[b!]
\renewcommand{\arraystretch}{1.2}
	\centerline{
	\begin{array}[t]{|c|ccc|}
	\hline
	      &y_1&y_2&y_3\\
	\hline 
	x_1   & (a,\gamma) & (b,\delta) & - \\
	x_2   & - & (a,\zeta) & (b,\gamma) \\
	x_3   & (b,\zeta) & - & (a,\delta) \\
	\hline
	\end{array}
	}
\caption{A pair of full-correlating subrectangles $\mathscr{C}_1$ and $\mathscr{C}_2$ of size 6. The first and second elements of the symbol pairs belong to $\mathscr{C}_1$ and $\mathscr{C}_2$, respectively.}
\label{fig:orthogonal_configs}
\end{figure}

\section{Link of Configurations and Latin Squares}\label{link_stopsets_latinSquares}

For investigating the occurrence or avoidance of stopping sets in a TD LDPC code (or, equivalently, of full configurations in a transversal design), we may consider the representation of these substructures in the underlying Latin squares.

\subsection{(Partial) Subrectangles}

\begin{definition}
Let $L$ be a Latin square with $n\times n$ cells. A (\emph{partial}) \emph{subrectangle} $\mathscr{C}$ of size $\ell$ is a subset of $\ell \leq n^2$ cells, explicitly given by the triples
\begin{multline*}
\mathscr{C}=\{(i_1,j_1,s_1),(i_2,j_2,s_2),\hdots,(i_\ell,j_\ell,s_\ell) :\\ L[i_k,j_k]=s_k,\ 1\leq k \leq \ell\}.
\end{multline*}
We say that $\mathscr{C}$ \emph{occurs in} $L$. Define $\mathcal{I}_\mathscr{C}=\bigcup i_k$ as the \emph{row set}, $\mathcal{J}_\mathscr{C}=\bigcup j_k$ as the \emph{column set} and $\mathcal{S}_\mathscr{C}=\bigcup s_k$ as the \emph{symbol set} of $\mathscr{C}$. A subrectangle is \emph{full}, if each element of $\mathcal{I}_\mathscr{C}$, $\mathcal{J}_\mathscr{C}$ and $\mathcal{S}_\mathscr{C}$ occurs in at least two triples of $\mathscr{C}$.
\end{definition}

Note that this concept has been similarly introduced in \cite{LaenMil07} for investigating stopping sets in codes from certain Steiner triple systems, which are also based on Latin squares. 
The term `partial' means that not every cell of the spanned subrectangle is filled. However, since a differentiation is not relevant for our purposes, we omit this expression for the rest of the paper.

\subsection{New Concept: Correlating Subrectangles in MOLS}

\begin{definition} Let $L_1,\hdots,L_m$ be MOLS. The subrectangles $\mathscr{C}_1,\hdots,\mathscr{C}_m$ are \emph{correlating}, if $\mathscr{C}_i$ occurs in $L_i$ and if the subrectangles cover the same cell positions within the Latin squares.
Additionally, if all subrectangles are full, then they are referred to as \mbox{\emph{full-correlating subrectangles}}. For any two correlating subrectangles $\mathscr{C}_i$ and $\mathscr{C}_j$, we say that they are
\begin{enumerate}
\item \emph{coincident}, since they cover the same cell positions, and
\item \emph{orthogonal}, i.e., for all $(s, s')\in \mathcal{S}_{\mathscr{C}_i} \times \mathcal{S}_{\mathscr{C}_j}$ it follows that
\[
|\{(x,y): (x,y,s) \in \mathscr{C}_i, (x,y,s')\in \mathscr{C}_j\}|\leq 1.
\]
\end{enumerate}
\end{definition}

\begin{example}
Fig. \ref{fig:orthogonal_configs} shows a pair of full-correlating subrectangles $\mathscr{C}_1$ and $\mathscr{C}_2$ of size 6, having a common row and column set $\mathcal{I}=\{x_1,x_2,x_3\}$ and $\mathcal{J}=\{y_1,y_2,y_3\}$, respectively, and symbol sets $\mathcal{S}_{\mathscr{C}_1}=\{a,b\}$ and $\mathcal{S}_{\mathscr{C}_2}=\{\gamma,\delta,\zeta\}$.
\end{example}

\begin{theorem}\label{link_configurations_to_latin_squares}

Let $L_1,...,L_m$ be MOLS of order $n$ and TD$(m+2,n)$ the associated transversal design. A collection of $m$ correlating subrectangles $\mathscr{C}_1,...,\mathscr{C}_m$ of size $\ell$ with common row and column sets $\mathcal{I}$ and $\mathcal{J}$, respectively, correspond to a $(q,l)$-configuration with $q = |\mathcal{I}|+|\mathcal{J}|+\sum_{1\leq i\leq \ell} |\mathcal{S}_{\mathscr{C}_i}|$ points and $l=\ell$ lines in the  TD$(m+2,n)$. If all subrectangles are full, they correspond to a full configuration in the transversal design and thus to a stopping set of size $\ell$ in the arising code.
\end{theorem}

\begin{proof}
This relationship can be verified by considering the equivalence between MOLS and transversal designs. 
\end{proof}

\subsection{Investigation of Full-Correlating Subrectangles in MOLS}

We are interested in determining full-correlating subrectangles in MOLS, since they correspond to full configurations in transversal designs and thus to stopping sets in the arising codes. Fig.~\ref{upto7cellconfigs} displays all possible full subrectangles of size $\ell\leq 7$ which might occur in a Latin square. Note that the order of the rows and columns are not relevant, since, despite any reordering, the structural dependencies of a subrectangle are retained. Now, we have a closer look at the following cases:

\subsubsection{Single Latin square ($m=1$)}

If a single Latin square is used for a TD$(3,n)$, every subrectangle depicted by figure \mbox{(a)-(n)} results directly in a stopping set, as long as the subrectangle occurs in the Latin square. Since the smallest subrectangle has size 4, we can bound the stopping distance of a TD LDPC code with column weight 3 by $s_{min}\geq 4$.

\subsubsection{Two orthogonal Latin squares ($m=2$)}

In this case, only a pair of full-correlating subrectangles lead to a stopping set in the arising code. Thus, we search for pairs of coincident and orthogonal subrectangles among the figures (a)-(n).  
The 4-square depicted by (a) has no coincident pair and thus can not induce a stopping set. The orthogonal pairs among the coincident subrectangles \mbox{(b)-(f)} of size 6 are \mbox{(b)+(c)}, \mbox{(d)+(e)}, \mbox{(d)+(f)} and \mbox{(e)+(f)}. The pair \mbox{(b)+(c)} results in the full \mbox{$(11,6)$-configuration} depicted by Fig. \ref{fig:configsBlockSize4}~(B) by using Theorem~\ref{link_configurations_to_latin_squares} and the remaining pairs in the star visualized by (C). Further coincident and orthogonal pairs of subrectangles are \mbox{(g)+(h}) and \mbox{(i)+(j)}, both resulting in the full $(11,6)$-configuration (B). Among the full subrectangles (k)-(n) of size~7, there are no orthogonal pairs. Note that the generalized Pasch-configuration depicted by (A) can not occur in the TD$(4,n)$, since there are no corresponding subrectangles.
The smallest full-correlating subrectangles, that might occur, have size 6 and thus, the stopping distance of a TD LDPC code with column weight 4 can be bounded by $s_{min}\geq 6$.

By designing the MOLS in such a way that some of the full-correlating subrectangles (in particular small ones) are avoided, we can prevent certain stopping sets from occurring in the arising codes. In the following section, we introduce a class of simple-structured MOLS which are ideally suited for this purpose.

\begin{figure}[!t]
\renewcommand{\thesubfigure}{\Alph{subfigure}}
\centering
	\subfloat[(10,5)]{
		\includegraphics[scale=0.42, trim=0 0 0 2cm]{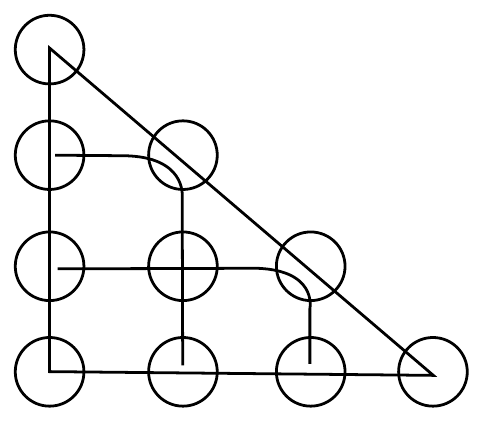}
	}
	\subfloat[(11,6)]{
		\includegraphics[scale=0.42, trim=0 0 0 2cm]{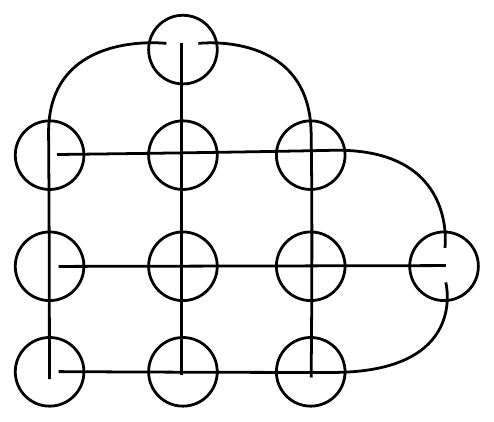}
	}
	\subfloat[(12,6)]{
		\includegraphics[scale=0.39, trim=0 0 0 2cm]{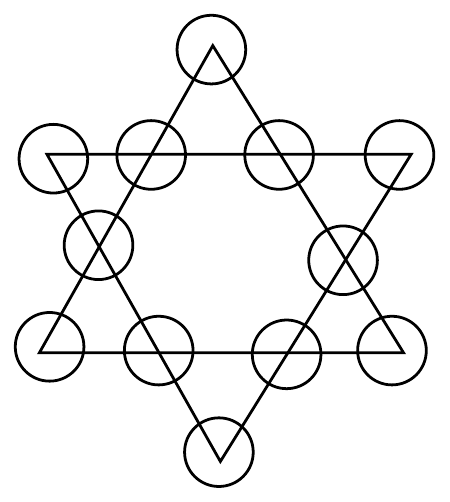}
	}
	\caption{Geometric representation of possible full $(q,l)$-configurations for $l\leq 6$ lines, 4 points per line and with the constraint that any two lines intersect in at most one point. The points and lines can be considered as the points and blocks of a combinatorial design with block size 4, where a block contains a point if the corresponding line goes through the point.}
\label{fig:configsBlockSize4}
\end{figure}

\section{TD LDPC Codes on Simple-Structured MOLS}\label{codes_on_simple_structured_MOLS}

We first need the following straightforward lemma, giving mutually orthogonal Latin squares isomorphic to Cayley addition tables (cf.~\cite{crc_handbook}):

\begin{lemma}\label{simple_structured_MOLS}
Let $\mathbb{F}_q$ be the Galois field of any prime power order $q$. We obtain a Latin square $\mathcal{L}^{(\alpha,\beta)}_q$ of order $q$ and scale factors $\alpha,\beta \in \mathbb{F}_{q}^{*}= \mathbb{F}_q\setminus\{0\}$ by
\[
\mathcal{L}^{(\alpha,\beta)}_q [x,y] = \alpha x + \beta y,\ x,y \in\mathbb{F}_q.
\]

If $\beta=1$, we simply write $\mathcal{L}^{(\alpha)}_q$ instead of $\mathcal{L}^{(\alpha,1)}_q$.  
Let $\mathbb{F}_q^*=\{\phi_1,\phi_2,\hdots,\phi_{q-1}\}$. 
Now, we associate every pair $(\alpha,\beta)\in (\mathbb{F}_q^*)^2$ with a Latin square $\mathcal{L}^{(\alpha,\beta)}_q$ and define equivalence classes by $U_i = \{x\odot(\phi_i,1) : x \in \mathbb{F}_q^* \}$ for $i=1,\hdots,q-1$, where $\odot$ is the elementwise multiplication over $\mathbb{F}_q$. These classes partition the set $(\mathbb{F}_q^*)^2$. Let $(\alpha_1,\beta_1),\hdots,(\alpha_{q-1},\beta_{q-1})$ be any representative system, i.e., $(\alpha_i,\beta_i)\in U_i$. Then, the Latin squares of any $m$-subset of $\{\mathcal{L}_q^{(\alpha_i,\beta_i)} : i=1,\hdots,q-1\}$, $1\leq m \leq q-1$,  are $m$ MOLS that can be used to build up a TD$(m+2,q)$ and thus to construct a TD LDPC code.
\end{lemma}

\begin{figure}[!b]
\renewcommand{\arraystretch}{1}
\setlength{\arraycolsep}{3.5pt}
\centering
\begin{tikzpicture}[xscale=2.2,yscale=2.3]
\path [draw] (0.55,0.45) rectangle (1.45,-3.55);
\path [draw] (1.55,0.45) rectangle (2.45,-3.55);
\path [draw] (2.55,0.45) rectangle (3.45,-3.55);
\path [draw] (3.55,0.45) rectangle (4.45,-3.55);

\draw[line width=0.5pt, dashed] (0.52,0.48) -- ++(3.96,0) -- ++(0,-1.05) -- ++(-3.96,0) -- ++(0,1.05);

\node[scale=0.8] at (1,0) {
		$\begin{array}[b]{|ccccc|}
		\hline
		0&1&2&3&4\\
		1&2&3&4&0\\
		2&3&4&0&1\\
		3&4&0&1&2\\
		4&0&1&2&3\\
		\hline
		\end{array}$
};
\node[scale=0.8] at (1,-0.46) {$\mathcal{L}^{(1,1)}_5$};

\node[scale=0.8] at (2,0) {
		$\begin{array}[b]{|ccccc|}
		\hline
		0&1&2&3&4\\
		2&3&4&0&1\\
		4&0&1&2&3\\
		1&2&3&4&0\\
		3&4&0&1&2\\
		\hline
		\end{array}$
};
\node[scale=0.8] at (2,-0.46) {$\mathcal{L}^{(2,1)}_5$};

\node[scale=0.8] at (3,0) {
		$\begin{array}[b]{|ccccc|}
		\hline
		0&1&2&3&4\\
		3&4&0&1&2\\
		1&2&3&4&0\\
		4&0&1&2&3\\
		2&3&4&0&1\\
		\hline
		\end{array}$
};
\node[scale=0.8] at (3,-0.46) {$\mathcal{L}^{(3,1)}_5$};

\node[scale=0.8] at (4,0) {
		$\begin{array}[b]{|ccccc|}
		\hline
		0&1&2&3&4\\
		4&0&1&2&3\\
		3&4&0&1&2\\
		2&3&4&0&1\\
		1&2&3&4&0\\
		\hline
		\end{array}$
};
\node[scale=0.8] at (4,-0.46) {$\mathcal{L}^{(4,1)}_5$};

\node[scale=0.8] at (1,-1) {
		$\begin{array}[b]{|ccccc|}
		\hline
		0&2&4&1&3\\
		2&4&1&3&0\\
		4&1&3&0&2\\
		1&3&0&2&4\\
		3&0&2&4&1\\
		\hline
		\end{array}$
};
\node[scale=0.8] at (1,-1.46) {$\mathcal{L}^{(2,2)}_5$};

\node[scale=0.8] at (2,-1) {
		$\begin{array}[b]{|ccccc|}
		\hline
		0&2&4&1&3\\
		4&1&3&0&2\\
		3&0&2&4&1\\
		2&4&1&3&0\\
		1&3&0&2&4\\
		\hline
		\end{array}$
};
\node[scale=0.8] at (2,-1.46) {$\mathcal{L}^{(4,2)}_5$};

\node[scale=0.8] at (3,-1) {
		$\begin{array}[b]{|ccccc|}
		\hline
		0&2&4&1&3\\
		1&3&0&2&4\\
		2&4&1&3&0\\
		3&0&2&4&1\\
		4&1&3&0&2\\
		\hline
		\end{array}$
};
\node[scale=0.8] at (3,-1.46) {$\mathcal{L}^{(1,2)}_5$};

\node[scale=0.8] at (4,-1) {
		$\begin{array}[b]{|ccccc|}
		\hline
		0&2&4&1&3\\
		3&0&2&4&1\\
		1&3&0&2&4\\
		4&1&3&0&2\\
		2&4&1&3&0\\
		\hline
		\end{array}$
};
\node[scale=0.8] at (4,-1.46) {$\mathcal{L}^{(3,2)}_5$};

\node[scale=0.8] at (1,-2) {
		$\begin{array}[b]{|ccccc|}
		\hline
		0&3&1&4&2\\
		3&1&4&2&0\\
		1&4&2&0&3\\
		4&2&0&3&1\\
		2&0&3&1&4\\
		\hline
		\end{array}$
};
\node[scale=0.8] at (1,-2.46) {$\mathcal{L}^{(3,3)}_5$};

\node[scale=0.8] at (2,-2) {
		$\begin{array}[b]{|ccccc|}
		\hline
		0&3&1&4&2\\
		1&4&2&0&3\\
		2&0&3&1&4\\
		3&1&4&2&0\\
		4&2&0&3&1\\
		\hline
		\end{array}$
};
\node[scale=0.8] at (2,-2.46) {$\mathcal{L}^{(1,3)}_5$};

\node[scale=0.8] at (3,-2) {
		$\begin{array}[b]{|ccccc|}
		\hline
		0&3&1&4&2\\
		4&2&0&3&1\\
		3&1&4&2&0\\
		2&0&3&1&4\\
		1&4&2&0&3\\
		\hline
		\end{array}$
};
\node[scale=0.8] at (3,-2.46) {$\mathcal{L}^{(4,3)}_5$};

\node[scale=0.8] at (4,-2) {
		$\begin{array}[b]{|ccccc|}
		\hline
		0&3&1&4&2\\
		2&0&3&1&4\\
		4&2&0&3&1\\
		1&4&2&0&3\\
		3&1&4&2&0\\
		\hline
		\end{array}$
};
\node[scale=0.8] at (4,-2.46) {$\mathcal{L}^{(2,3)}_5$};

\node[scale=0.8] at (1,-3) {
		$\begin{array}[b]{|ccccc|}
		\hline
		0&4&3&2&1\\
		4&3&2&1&0\\
		3&2&1&0&4\\
		2&1&0&4&3\\
		1&0&4&3&2\\
		\hline
		\end{array}$
};
\node[scale=0.8] at (1,-3.46) {$\mathcal{L}^{(4,4)}_5$};

\node[scale=0.8] at (2,-3) {
		$\begin{array}[b]{|ccccc|}
		\hline
		0&4&3&2&1\\
		3&2&1&0&4\\
		1&0&4&3&2\\
		4&3&2&1&0\\
		2&1&0&4&3\\
		\hline
		\end{array}$
};
\node[scale=0.8] at (2,-3.46) {$\mathcal{L}^{(3,4)}_5$};

\node[scale=0.8] at (3,-3) {
		$\begin{array}[b]{|ccccc|}
		\hline
		0&4&3&2&1\\
		2&1&0&4&3\\
		4&3&2&1&0\\
		1&0&4&3&2\\
		3&2&1&0&4\\
		\hline
		\end{array}$
};
\node[scale=0.8] at (3,-3.46) {$\mathcal{L}^{(2,4)}_5$};

\node[scale=0.8] at (4,-3) {
		$\begin{array}[b]{|ccccc|}
		\hline
		0&4&3&2&1\\
		1&0&4&3&2\\
		2&1&0&4&3\\
		3&2&1&0&4\\
		4&3&2&1&0\\
		\hline
		\end{array}$
};
\node[scale=0.8] at (4,-3.46) {$\mathcal{L}^{(1,4)}_5$};
\end{tikzpicture}
\caption{The solid-line boxes represent the Latin squares associated with the equivalence classes $U_1,\hdots,U_4$ over $(\mathbb{F}_5^*)^2$. The Latin squares of the first row, bordered by a dashed line, constitute a possible representative system of these classes and thus are MOLS.}
\label{equivclasses}
\end{figure}

\begin{example}
Fig. \ref{equivclasses} depicts the four equivalence classes over $(\mathbb{F}_5^*)^2$, where the members of the equivalence classes are represented by the associated Latin squares. The Latin squares of the first row, bordered by a dashed line, constitute a possible representative system of these classes and thus are MOLS.
\end{example}

\begin{lemma}\label{replacing_latin_squares}
Let $\mathcal{L}^{(\alpha_1,\beta_1)}_q,\hdots,\mathcal{L}^{(\alpha_m,\beta_m)}_q$ be MOLS, i.e., $(\alpha_i,\beta_i)$ are from different equivalence classes,  $1\leq i \leq m$. We may replace $\mathcal{L}^{(\alpha_i,\beta_i)}_q$ by any $\mathcal{L}^{(\alpha'_i,\beta'_i)}_q$ without changing the stopping set distribution of the resulting TD LDPC code if $(\alpha'_i,\beta'_i)$ is in the same equivalence class as $(\alpha_i,\beta_i)$.
\end{lemma}
\begin{proof}
There is a bijection between the symbol sets of $\mathcal{L}^{(\alpha_i,\beta_i)}_q$ and $\mathcal{L}^{(\alpha'_i,\beta'_i)}_q$ that preserves the structure of the Latin squares (renaming of the symbols). By applying this symbol-renaming to any subrectangle, we obtain a bijective mapping between the subrectangles of both Latin squares, leading to the same stopping set distributions.
\end{proof}

\begin{theorem}\label{reduced_form}
The MOLS $\mathcal{L}^{(\alpha_1,\beta_1)}_q,\hdots,\mathcal{L}^{(\alpha_m,\beta_m)}_q$ can be reduced to the form $\mathcal{L}^{(\alpha'_1,1)}_q,\hdots,\mathcal{L}^{(\alpha'_m,1)}_q$ by $\alpha'_i:=\alpha_i\beta_i^{-1}$ over $\mathbb{F}_q$, $1\leq i\leq m$, without affecting the stopping set distribution.
\end{theorem}
\begin{proof}
Since $(\alpha_i,\beta_i)=\beta_i\odot(\alpha'_i,1)$, both pairs are in the same equivalence class and thus $\mathcal{L}^{(\alpha_i,\beta_i)}_q$ can be replaced by $\mathcal{L}^{(\alpha_i\beta_i^{-1},1)}_q$ according to Lemma \ref{replacing_latin_squares}.
\end{proof}

As a consequence of Theorem \ref{reduced_form}, we may simplify our considerations to the representative system $\{(\alpha,1):\alpha\in\mathbb{F}_q^*\}$ and the associated MOLS $\{\mathcal{L}^{(\alpha,1)}_q : \ 1\leq\alpha\leq q-1\}$. Although this restriction is reasonable in order to improve the stopping set distributions (Section \ref{improved_codes}) of the arising codes, we can achieve low-complexity encoding by the appropriate choice of another representative system (Section~\ref{encoding}).

\begin{definition}
An \emph{$\ell$-polygon} $\mathscr{P}$ in the Latin square $\mathcal{L}^{(\alpha,\beta)}_q$ is a special subrectangle of size $\ell$ (with $\ell$ even) of the form
\begin{multline*}
\mathscr{P}=\{(i_1,j_1,s_1), (i_2,j_2,s_2), \hdots, (i_\ell,j_\ell,s_\ell) : \\
i_1=i_2,\: j_2=j_3,\: i_3=i_4,\hdots, i_{\ell-1}=i_\ell,\: j_{\ell}=j_1,\\ \mbox{and } \mathcal{L}^{(\alpha,\beta)}_q[i_k,j_k]=s_k,\ 1\leq k \leq \ell\}.
\end{multline*}
\end{definition}

\begin{theorem}
\label{polygon_condition}
For an $\ell$-polygon $\mathscr{P}$ in the Latin square $\mathcal{L}^{(\alpha,\beta)}_q$ with triples $(i_1,j_1,s_1),\hdots,(i_\ell,j_\ell,s_\ell)$ it must hold that
\[
\sum\limits_{1\leq t\leq \ell}(-1)^{t-1} s_t = 0\ \ \mbox{(over } \mathbb{F}_q \mbox{)}.
\]
\end{theorem}

\begin{proof}
From $s_t=\alpha i_t + \beta j_t$ (Lemma~\ref{simple_structured_MOLS}) it follows that
\begin{multline*}
\sum\limits_{t}(-1)^{t-1} s_t = \alpha\left(\sum\limits_{t}(-1)^{t-1} i_t\right)+\beta\left(\sum\limits_{t}(-1)^{t-1} j_t\right)\\
= \alpha(\underbrace{i_1-i_2}_{=0}+\underbrace{i_3-i_4}_{=0}+\hdots+\underbrace{i_{\ell-1}-i_\ell}_{=0})\\+\beta(\underbrace{-j_2+j_3}_{=0}\underbrace{-j_3+j_4}_{=0}-\hdots\underbrace{-j_\ell+j_1}_{=0})=0.
\end{multline*}
\end{proof}

\subsection{Properties of $\mathscr{L}^m_q$-TD LDPC codes}

\begin{definition}
Define $\mathscr{L}^1_q=\{\mathcal{L}^{(\alpha,\beta)}_q : (\alpha,\beta)\in(\mathbb{F}_q^*)^2\}$ and let $\mathscr{L}_q^m$ be the collection of all ordered $m$-subsets of $\mathscr{L}^1_q$ that are MOLS. Each transversal design arising from an element of $\mathscr{L}_q^m$ is referred to as an $\mathscr{L}^m_q$-TD and the resulting code as an $\mathscr{L}^m_q$-TD LDPC code. Note that the order of the Latin squares within a set of MOLS is irrelevant, since every order leads to the same code. However, for a clearer presentation of the paper we distinguish between these cases and consider only ordered sets.
\end{definition}

An $\mathscr{L}^m_q$-TD LDPC code has block length $N=q^2$, rate $R\geq (q-m-2)/q$, and the code's parity-check matrix is regular with column weight $m+2$ and row weight $q$. Furthermore, a subclass of these codes has quasi-cyclic structure (cf. Section~\ref{encoding}), leading to a low encoding complexity linear in the block length. The parity-check matrix of a  quasi-cyclic $\mathscr{L}^m_q$-TD LDPC code consists of $(m+2)\times q$ circulant submatrices (called circulants) of size $q\times q$. For a more flexible code design, we can use an arbritrary submatrix of $(m+2)\times a$ circulants with $1\leq a \leq q$ as a parity-check matrix, leading to a code of block length $N=aq$, rate $R\geq (a-m-2)/a$, column weight $m+2$ and row weight $a$. Hence, by varying $a$, we obtain a wide spectrum of block lengths and code rates for any prime power order $q$.

\subsection{Stopping Distance of $\mathscr{L}^1_q$-TD LDPC codes}

Let $\mathbb{F}_q$ be the Galois field of order $q$ and characteristic $\omega_q$ (i.e., $q=\omega_q^x$ for any integer $x\geq 1$), and let
$\mathcal{L}^{(\alpha)}_q\in \mathscr{L}^1_q$ be the Latin square with scale factors $\alpha\in \mathbb{F}^*_q$ and $\beta = 1$, which leads to an  $\mathscr{L}^1_q$-TD LDPC code of column weight 3. 
We now investigate the occurrence of small subrectangles in $\mathcal{L}^{(\alpha)}_q$, which result in small stopping sets. For this, we first consider all possible full subrectangles up to size $7$, displayed by Fig.~\ref{upto7cellconfigs} (where $a,b,c\in\mathbb{F}_q$), that might occur in any Latin square. 
The smallest full subrectangle is a square of four cells as depicted by (a).
Since it is a \mbox{4-polygon}, it must hold that $2a-2b=0$ (over $\mathbb{F}_q$) according to Theorem~\ref{polygon_condition}. It follows that $a=b$ for $\omega_q\neq 2$, which contradicts the properties of a Latin square. Thus, the 4-square can not occur in $\mathcal{L}^{(\alpha)}_q$ if $\omega_q\neq 2$, else (if $\omega_q = 2$) it occurs.
Fig.~\ref{upto7cellconfigs} \mbox{(b)-(j)} depicts all possible full subrectangles of size 6, of which the figures \mbox{(b)-(f)} form a 6-polygon.
With Theorem~\ref{polygon_condition}, we can exclude \mbox{(c)-(f)} from occurring in the Latin square $\mathcal{L}^{(\alpha)}_q$ if $\omega_q\neq 2,3$. For example, it holds for (c) that 
$3(a-b)= 0$ and thus $a=b$ for $\omega_q\neq 3$, a contradiction. 
Fig.~\ref{upto7cellconfigs}~(b) occurs in $\mathcal{L}^{(\alpha)}_q$ and thus results in a stopping set of size 6. An explicit subrectangle, which is isomorphic to (b), is illustrated by the left table of Fig.~\ref{fig:6cycle_config}.
Each of the rectangles \mbox{(g)-(j)} contain three 4-polygons and thus, a linear system of three equations and three variables must be satisfied. The solutions of these systems lead to invalid subrectangles if $\omega_q>3$ and thus can not occur in $\mathcal{L}^{(\alpha)}_q$ for $\omega_q>3$. 
The only occurring full subrectangle of size~7 is given by Fig.~\ref{upto7cellconfigs}~(k). An example is visualized by the right table of Fig.~\ref{fig:6cycle_config}, which is isomorphic to (k).
As a conclusion, we can avoid certain small subrectangles in $\mathcal{L}^{(\alpha)}_q$ (which correspond to harmful stopping sets), if $\omega_q>3$.

\begin{figure}[t!]
\renewcommand{\arraystretch}{1.2}
\centering
\subfloat{
	\begin{array}[b]{|c|ccc|}
	\hline
	      &0&\alpha&2\alpha\\
	\hline
	0   & - & \alpha & 2\alpha  \\
	1   & \alpha & - & 3\alpha \\
	2   & 2\alpha & 3\alpha  & -\\
	\hline
	\end{array}
}
\hspace{10pt}
\subfloat{
	\begin{array}[b]{|c|ccc|}
	\hline
	      &0&\alpha&2\alpha\\
	\hline
	0   & - & \alpha & 2\alpha  \\
	1   & \alpha & 2\alpha & 3\alpha \\
	2   & 2\alpha & 3\alpha  & -\\
	\hline
	\end{array}
}
\caption{Full subrectangles in $\mathcal{L}^{(\alpha)}_q$ of size 6 and 7}
\label{fig:6cycle_config}
\end{figure}

\begin{theorem}\label{stopdist_weight3}
For an $\mathscr{L}^1_{q}$-TD LDPC code, the stopping distance $s_{min}$ is given by
\[s_{min}=\begin{cases}
  4,  & \text{if } \omega_q=2,\\
  6, & \text{if } \omega_q>2,
\end{cases}\]
where $\omega_q$ is the characteristic of $\mathbb{F}_q$ (which must be prime).
\end{theorem}

\begin{proof}
As shown above, the size of the smallest full subrectangle for $\omega_q=2$ and $\omega_q>2$ is 4 and 6, respectively, which corresponds to the size of the smallest stopping set. 
\end{proof}

\subsection{Bounding the Stopping Distance of $\mathscr{L}^2_q$-TD LDPC codes}

Let $\mathbb{F}_q$ be the Galois field of order $q$ and characteristic $\omega_q$, and let $(\mathcal{L}^{(\alpha_1)}_q,\mathcal{L}^{(\alpha_2)}_q)\in\mathscr{L}^2_q$ be two orthogonal Latin squares with scale factors $(\alpha_1,1)$ and $(\alpha_2,1)$, respectively, where $\alpha_1,\alpha_2 \in \mathbb{F}^*_q$ and $\beta_1=\beta_2 =1$. These MOLS lead to an $\mathscr{L}^2_q$-TD LDPC code of column weight 4. 
The smallest full subrectangle of four cells, depicted by Fig.~\ref{upto7cellconfigs}~(a), does not have a coincident pair and thus can not lead to a stopping set of size 4. Also, there is no stopping set of size $5$, since there are no full subrectangles of this size. 
If $\omega_q>3$, the only full subrectangle of size 6 that occurs in the single Latin square $\mathcal{L}^{(\alpha_1)}_q$ is Fig.~\ref{upto7cellconfigs}~(b). To result in a stopping set of the $\mathscr{L}^2_q$-TD LDPC code, there must be a full-correlating subrectangle. The only coincident and orthogonal pair of (b) is (c) (which is also full), but this can not occur in $\mathcal{L}^{(\alpha_2)}_q$.
Among the coincident full subrectangles of size~7, depicted by Fig.~\ref{upto7cellconfigs} \mbox{(k)-(n)}, no orthogonal pairs exist and thus no stopping sets of size 7.

\begin{theorem}\label{stopdist_weight4}
For an $\mathscr{L}^2_q$-TD LDPC code, the stopping distance $s_{min}$ can be bounded by
\[s_{min}\geq \begin{cases}
  6,  & \text{if } \omega_q=2 \text{ or } 3,\\
  8, & \text{if } \omega_q>3.
\end{cases}\]
\end{theorem}

\begin{proof}
As shown above, there is no stopping set of size $l < 6$ and thus, $s_{min}\geq 6$ must hold in general. If $\omega_q>3$, we can exclude further pairs of full-correlating subrectangles up to size 8 and thus, there can not be any stopping set of size $l < 8$, giving $s_{min}\geq 8$.
\end{proof}

\section{TD LDPC Codes with Improved\\Stopping Set Distributions}\label{improved_codes}

\subsection{Avoidance of Small Stopping Sets in $\mathscr{L}^2_q$-TD LDPC codes}\label{avoiding_small_stopsets}

The stopping distance of any column-weight-4 $\mathscr{L}^2_q$-TD LDPC code is \mbox{$s_{min} \geq 8$}  for $\omega_q>3$ (where $q$ is the characteristic of $\mathbb{F}_q$) and thus, the smallest stopping sets may have size 8. By simulations, we revealed only two types of full-correlating subrectangles of this size (Fig.~\ref{stopsets_size_8}) that may occur in any Latin square of $\mathscr{L}^2_q$ for $\omega_q>3$. Let $(\mathcal{L}^{(\alpha_1)}_q,\mathcal{L}^{(\alpha_2)}_q)\in\mathscr{L}^2_q$. We claim that by an appropriate choice of $\alpha_1$ and $\alpha_2$ both substructures can be avoided. We demonstrate this for the first type of Fig.~\ref{stopsets_size_8}, given by the full-correlating subrectangles $\mathscr{C}_1=\{(x_1,y_1,a),(x_1,y_2,b),\hdots,(x_4,y_4,d)\}$ in $\mathcal{L}^{(\alpha_1)}_q$ and $\mathscr{C}_2=\{(x_1,y_1,\gamma),(x_1,y_2,\delta),\hdots,(x_4,y_4,\delta)\}$ in $\mathcal{L}^{(\alpha_2)}_q$.

For each triple $(x_i, y_j, s)$ of $\mathscr{C}_1$, it holds that $\alpha_1 x_i + y_j = s$  (over $\mathbb{F}_q$) according to Lemma \ref{simple_structured_MOLS}. Hence, for every two triples $(x_i, y_j, s)$ and $(x_k, y_l, s)$ with the same symbol $s$ we can set up an equation $\alpha_1 x_i + y_j - \alpha_1 x_k - y_l = 0$. Analogously, we can derive equations for $\mathscr{C}_2$. Altogether, we obtain the following system of linear equations (over $\mathbb{F}_q$):
\begin{eqnarray*}
\alpha_1 x_1 + y_1 - \alpha_1 x_3 - y_4 = 0,\ \ \ \alpha_2 x_1 + y_1 - \alpha_2 x_2 - y_3 = 0,\\
\alpha_1 x_1 + y_2 - \alpha_1 x_3 - y_3 = 0,\ \ \ \alpha_2 x_1 + y_2 - \alpha_2 x_4 - y_4 = 0,\\
\alpha_1 x_2 + y_2 - \alpha_1 x_4 - y_1 = 0,\ \ \ \alpha_2 x_2 + y_2 - \alpha_2 x_3 - y_4 = 0,\\
\alpha_1 x_2 + y_3 - \alpha_1 x_4 - y_4 = 0,\ \ \ \alpha_2 x_3 + y_3 - \alpha_2 x_4 - y_1 = 0.\,
\end{eqnarray*}
By Gaussian elimination, we obtain $(2\alpha_1-\alpha_2)(y_3-y_4) = 0$ and thus $y_3 = y_4$ or $2\alpha_1-\alpha_2 = 0$. Since $y_3$ and $y_4$ represent columns, they must be distinct. The second condition can be prevented from being satisfied by a proper choice of $\alpha_1$ and $\alpha_2$. More precisely, the given substructure can not occur in $(\mathcal{L}^{(\alpha_1)}_q,\mathcal{L}^{(\alpha_2)}_q)\in\mathscr{L}^2_q$, if
\begin{align*}
2 \alpha_1 - \alpha_2 & \neq 0, \mbox{ and} \tag{C1}\\
2 \alpha_2 - \alpha_1 & \neq 0\ \ (\mbox{over }\mathbb{F}_q). \tag{C2}
\end{align*}
The constraint C2 must be satisfied, since the roles of $\alpha_1$ and $\alpha_2$ can be interchanged.
Analogously, the second type of full-correlating subrectangles (Fig.~\ref{stopsets_size_8}) can be avoided, if
\begin{align*}
\alpha_1 + \alpha_2 & \neq 0\ (\mbox{over }\mathbb{F}_q). \tag{C3}
\end{align*}

The simulations also showed that there does not occur any stopping set of size 9 for $\omega_q > 3$. 
As a consequence, since the stopping sets of size 8 can be avoided, we have raised the stopping distance from 8 to 10 by a proper choice of the scale factors for $\omega_q > 3$ .

\begin{figure}[t!]
\setlength{\arraycolsep}{2.5pt}
\renewcommand{\arraystretch}{1.2}
\subfloat{
	\scalebox{0.95}{$
	\begin{array}[t]{|c|cccc|}
	\hline
	      &y_1&y_2&y_3&y_4\\
	\hline 
	x_1   & (a,\gamma) & (b,\delta) & - & -\\
	x_2   & - & (c,\zeta) & (d,\gamma) & -\\
	x_3   & - & - & (b,\eta) & (a,\zeta)\\
	x_4   & (c,\eta) & - & - & (d,\delta)\\
	\hline
	\end{array}
	$}
}
\subfloat{
	\scalebox{0.95}{$
	\begin{array}[t]{|c|cccc|}
	\hline
	      &y_1&y_2&y_3&y_4\\
	\hline 
	x_1   & (a,\gamma) & (b,\delta) & - & -\\
	x_2   & (c,\zeta) & (d,\eta) & - & -\\
	x_3   & - & - & (b,\zeta) & (d,\gamma)\\
	x_4   & - & - & (a,\eta) & (c,\delta)\\
	\hline
	\end{array}
	$}
}
\caption{Two types of full-correlating subrectangles of size 8 in $\mathscr{L}^2_q$. The first elements of the symbol pairs belong to a subrectangle with symbol set $\{a,b,c,d\}$ and the second elements to a full-correlating counterpart with symbol set $\{\gamma,\delta,\zeta,\eta\}$.}
\label{stopsets_size_8}
\end{figure}

\definecolor{light-gray}{gray}{0.8}
\newcommand{\gr}[1]{\textcolor{light-gray}{#1}}

\captionsetup[subfloat]{margin=10pt,format=hang}

\begin{figure*}[t!]
\renewcommand{\thesubfigure}{\alph{subfigure}}
\setlength{\arraycolsep}{3pt}
\centerline{
\subfloat[]{
\renewcommand{\arraystretch}{1.2}
		\begin{array}[b]{|lllllll|}
		\hline
		\:\gr{0}\ &\gr{1}\ &\gr{2}\ &\gr{3}\ &\gr{4}\ &\gr{5}\ &\gr{6}\:\\
		\:\gr{1}&\gr{2}&\gr{3}&\gr{4}&5^*&6^*&\gr{0}\\
		\:\gr{2}&\gr{3}&\textbf{4}&\textbf{5}^*&\gr{6}&0^*&\gr{1}\\
		\:\gr{3}&\textbf{4}&\gr{5}&\textbf{6}^*&0^*&\gr{1}&\gr{2}\\
		\:\gr{4}&\textbf{5}&\textbf{6}&\gr{0}&\gr{1}&\gr{2}&\gr{3}\\
		\:\gr{5}&\gr{6}&\gr{0}&\gr{1}&\gr{2}&\gr{3}&\gr{4}\\
		\:\gr{6}&\gr{0}&\gr{1}&\gr{2}&\gr{3}&\gr{4}&\gr{5}\\
		\hline
		\end{array}\begin{array}[b]{|lllllll|}
		\hline
		\:\gr{0}\ &\gr{1}\ &\gr{2}\ &\gr{3}\ &\gr{4}\ &\gr{5}\ &\gr{6}\:\\
		\:\gr{2}&\gr{3}&\gr{4}&\gr{5}&6^*&0^*&\gr{1}\\
		\:\gr{4}&\gr{5}&\textbf{6}&\textbf{0}^*&\gr{1}&2^*&\gr{3}\\
		\:\gr{6}&\textbf{0}&\gr{1}&\textbf{2}^*&3^*&\gr{4}&\gr{5}\\
		\:\gr{1}&\textbf{2}&\textbf{3}&\gr{4}&\gr{5}&\gr{6}&\gr{0}\\
		\:\gr{3}&\gr{4}&\gr{5}&\gr{6}&\gr{0}&\gr{1}&\gr{2}\\
		\:\gr{5}&\gr{6}&\gr{0}&\gr{1}&\gr{2}&\gr{3}&\gr{4}\\
		\hline
		\end{array}
}
\hspace{20pt}
\subfloat[]{
\renewcommand{\arraystretch}{1.2}
		\begin{array}[b]{|lllllll|}
		\hline
		\:\gr{0}\ &\gr{1}\ &\gr{2}\ &\gr{3}\ &\gr{4}\ &\gr{5}\ &\gr{6}\:\\
		\:\gr{1}&\gr{2}&\textbf{3}&\textbf{4}&\gr{5}&\gr{6}&\gr{0}\\
		\:\gr{2}&\gr{3}&\gr{4}&\gr{5}&\gr{6}&\gr{0}&\gr{1}\\
		\:\textbf{3}&\gr{4}&\gr{5}&\textbf{6}&\gr{0}&1^*&2^*\\
		\:\textbf{4}&\gr{5}&\textbf{6}&\gr{0}&\gr{1}&\gr{2}&\gr{3}\\
		\:\gr{5}&\gr{6}&\gr{0}&1^*&\gr{2}&\gr{3}&4^*\\
		\:\gr{6}&\gr{0}&\gr{1}&2^*&\gr{3}&4^*&\gr{5}\\
		\hline
		\end{array}\begin{array}[b]{|lllllll|}
		\hline
		\:\gr{0}\ &\gr{1}\ &\gr{2}\ &\gr{3}\ &\gr{4}\ &\gr{5}\ &\gr{6}\:\\
		\:\gr{2}&\gr{3}&\textbf{4}&\textbf{5}&\gr{6}&\gr{0}&\gr{1}\\
		\:\gr{4}&\gr{5}&\gr{6}&\gr{0}&\gr{1}&\gr{2}&\gr{3}\\
		\:\textbf{6}&\gr{0}&\gr{1}&\textbf{2}&\gr{3}&4^*&5^*\\
		\:\textbf{1}&\gr{2}&\textbf{3}&\gr{4}&\gr{5}&\gr{6}&\gr{0}\\
		\:\gr{3}&\gr{4}&\gr{5}&6^*&\gr{0}&\gr{1}&2^*\\
		\:\gr{5}&\gr{6}&\gr{0}&1^*&\gr{2}&3^*&\gr{4}\\
		\hline
		\end{array}
}
}
\caption{Full-correlating subrectangles in $(\mathcal{L}^{(1)}_7,\mathcal{L}^{(2)}_7)\in\mathscr{L}^2_7$ composed of small subrectangles of size 6.}
\label{fig:example_config_extension}
\end{figure*}

\subsection{Obtaining Full-Correlating Subrectangles in $\mathscr{L}^2_q$}

In this subsection, we describe a technique for constructing full-correlating subrectangles in MOLS which correspond to stopping sets that are harmful for the decoding performance of the arising $\mathscr{L}^2_q$-TD LDPC codes. First, we introduce the concept of translations in Definition~\ref{translations}. We then build up full-correlating subrectangles in $\mathscr{L}^2_q$ by the composition of translations in Theorem~\ref{duplicating_configurations}. In Subsection~\ref{analytical optimization}, we finally show how to derive scale factors based on these subrectangles in order to improve the stopping set distribution.

\begin{definition}\label{translations}
Let $\mathscr{C}$ be a subrectangle \mbox{in $\mathcal{L}^{(\alpha)}_q$}. For any pair $(i,j)\in \mathbb{F}_q^2$, we obtain a \emph{translation}
\mbox{$\mathscr{C}+(i,j)$} \mbox{in $\mathcal{L}^{(\alpha)}_q$} by
\[
\mathscr{C}+(i,j):=\{(x+i,\ y+j,\ s+\alpha i + j) : (x,y,s)\in \mathscr{C}\}
\]
over $\mathbb{F}_q$. Clearly, $\mathscr{C}+(i,j)$ is full if $\mathscr{C}$ is full.
\end{definition}

\begin{lemma}\label{lemmaTranslation}
A translation $\mathscr{C}+(i,j)$ has the same symbols as $\mathscr{C}$ at the displaced positions if $\alpha i + j = 0$. 
\end{lemma}
\begin{proof}
This results directly from Definition~\ref{translations}, since every triple $ (x,y,s)\in \mathscr{C}$ leads to a triple $(x+i,y+i,s)\in \mathscr{C}+(i,j)$ with the same symbol $s$.
\end{proof}

\begin{theorem}\label{duplicating_configurations}
Let $(\mathcal{L}^{(\alpha_1)}_q,\mathcal{L}^{(\alpha_2)}_q)\in\mathscr{L}^2_q$ be two MOLS, and let $\mathscr{C}_1$ be any full subrectangle of \mbox{size $\ell$} \mbox{in $\mathcal{L}^{(\alpha_1)}_q$} and $\mathscr{C}_2$ the correlating (generally not full) subrectangle \mbox{in $\mathcal{L}^{(\alpha_2)}_q$} of the same size.
Then, we obtain a pair of full-correlating subrectangles $(\Psi_1,\Psi_2)$ with $\Psi_1 = \mathscr{C}_1 \cup (\mathscr{C}_1+(i,j))$ \mbox{in $\mathcal{L}^{(\alpha_1)}_q$} and 
$\Psi_2 = \mathscr{C}_2 \cup (\mathscr{C}_2+(i,j))$ \mbox{in $\mathcal{L}^{(\alpha_2)}_q$} for any $(i,j)\in\mathbb{F}^2_q\setminus\{(0,0)\}$ with  $i\alpha_2 + j = 0$ (over $\mathbb{F}_q$).
The size $\ell'$ of $\Psi_1$ and $\Psi_2$ can be bounded by $\ell+|\mathcal{S}^*_{\mathscr{C}_2}|\leq \ell' \leq 2\ell$, where $\mathcal{S}^*_{\mathscr{C}_2}=\{s\in \mathcal{S}_{\mathscr{C}_2} : s  \text{ occurs exactly once in the triples of } \mathscr{C}_2\}$.
\end{theorem}

\begin{proof}
The complete proof has been omitted, since it is very technical. The main idea is to translate $\mathscr{C}_2$ (which is not full due to some unique symbols) in such a way that the translation $\mathscr{C}_2+(i,j)$ has the same symbols at the displaced positions and thus doubles the unique symbols $\mathcal{S}^*_{\mathscr{C}_2}$ which occur only once in $\mathscr{C}_2$. The translation achieves this doubling, if $i\alpha_2 + j = 0$ (Lemma~\ref{lemmaTranslation}), even if the translation overlaps with $\mathscr{C}_2$ at several cell positions. It can easily be verified that there can not be a shared cell holding a symbol $s\in\mathcal{S}^*_{\mathscr{C}_2}$ and thus, any unique symbol lead to a new cell doubling this symbol. A better understanding should be gained by considering the subsequent example. 
\end{proof}

\begin{example}
Fig. \ref{fig:example_config_extension} (a) and (b) show full-correlating pairs of subrectangles of size $\ell ' = 10$ and $\ell ' = 12$, respectively, which are composed of correlating subrectangles $(\mathscr{C}_1,\mathscr{C}_2)$ of size 6 (marked in bold) and their translations of the same size (marked with an asterisk).
In Fig.~\ref{fig:example_config_extension}~(a), the unique symbols $\mathcal{S}^*_{\mathscr{C}_2}=\{3,6\}$ are doubled by the translation of  $\mathscr{C}_2$, leading to a full subrectangle. Since $\mathscr{C}_1$ and $\mathscr{C}_2$ overlap with their translations at two positions, we obtain 
full-correlating subrectangles of size 10 which is smaller than the upper bound $2\ell=12$ given in Theorem \ref{duplicating_configurations}.
By contrast, Fig.~\ref{fig:example_config_extension}~(b) demonstrates that the upper bound can be reached, since $\ell'=2\ell=12$. Every symbol occurring in $\mathscr{C}_2$ is unique, i.e., $\mathcal{S}_{\mathscr{C}_2}=\mathcal{S}^*_{\mathscr{C}_2}=\{1,\hdots,6\}$, and thus each of these symbols is guaranteed to be doubled by a separate cell. 
\end{example}

\begin{figure}[b!]
\renewcommand{\arraystretch}{1.5}
\renewcommand{\thesubfigure}{\alph{subfigure}}
\subfloat[Variable representation of all 6-polygons in $\mathcal{L}^{(\alpha_1)}_q$]{
	\scalebox{0.8}{$
	\begin{array}[t]{|c|ccc|}
	\hline
	        & y_1 & -\alpha_1 x_2 + \alpha_1 x_3 + y_1 & \alpha_1 x_1 - \alpha_1 x_2 + y_1\\
	\hline
	x_1   & \alpha_1 x_1 + y_1  & \alpha_1 x_1 - \alpha_1 x_2 + \alpha_1 x_3 + y_1 & - \\
	x_2   & - & \alpha_1 x_3 + y_1  & \alpha_1 x_1 + y_1 \\
	x_3   & \alpha_1 x_3 + y_1 & -  &  \alpha_1 x_1 - \alpha_1 x_2 + \alpha_1 x_3 + y_1\\
	\hline
	\end{array}
	$}
}

\subfloat[Variable representation of all correlating pairs of (a) in $\mathcal{L}^{(\alpha_2)}_q$]{
	\scalebox{0.8}{$
	\begin{array}[t]{|c|ccc|}
	\hline
	        & y_1 & -\alpha_1 x_2 + \alpha_1 x_3 + y_1 & \alpha_1 x_1 - \alpha_1 x_2 + y_1\\
	\hline
	x_1   & \alpha_2 x_1 + y_1  & \alpha_2 x_1 - \alpha_1 x_2 + \alpha_1 x_3 + y_1 & -  \\
	x_2   & - & \alpha_2 x_2 - \alpha_1 x_2 + \alpha_1 x_3 + y_1  & \alpha_2 x_2 + \alpha_1 x_1 -\alpha x_2 + y_1 \\
	x_3   & \alpha_2 x_3 + y_1 & -  &  \alpha_2 x_3 + \alpha_1 x_1 - \alpha_1 x_2 + y_1\\
	\hline
	\end{array}
	$}
}

\setcounter{subfigure}{0}
\renewcommand{\thesubfigure}{\roman{subfigure}}
\renewcommand{\arraystretch}{1}
\setlength{\arraycolsep}{2pt}
\subfloat[]{
	\begin{array}[b]{|ccc|}
	\hline
	\circledast & * & -\\
	- & \circledast & *\\
	* & - & *\\
	\hline
	\end{array}
}
\subfloat[]{
	\begin{array}[b]{|ccc|}
	\hline
	\circledast & * & -\\
	- & * & *\\
	* & - & \circledast\\
	\hline
	\end{array}
}
\subfloat[]{
	\begin{array}[b]{|ccc|}
	\hline
	* & \circledast & -\\
	- & * & *\\
	\circledast & - & *\\
	\hline
	\end{array}
}
\subfloat[]{
	\begin{array}[b]{|ccc|}
	\hline
	* & * & -\\
	- & * & \circledast\\
	\circledast & - & *\\
	\hline
	\end{array}
}
\subfloat[]{
	\begin{array}[b]{|ccc|}
	\hline
	* & \circledast & -\\
	- & * & \circledast\\
	* & - & *\\
	\hline
	\end{array}
}
\subfloat[]{
	\begin{array}[b]{|ccc|}
	\hline
	* & * & -\\
	- & \circledast & *\\
	* & - & \circledast\\
	\hline
	\end{array}
}
\caption{Subfigures (a) and (b) represent correlating pairs of subrectangles which are 6-polygons. Subfigures (i)-(vi) show the pairs of entries that might be equal in the subrectangles represented by (b).}
\label{fig:representation_6sided_polygon}
\end{figure}

\subsection{Analytical Optimization of $\mathscr{L}^2_q$-TD LDPC codes}\label{analytical optimization}

The full-correlating pairs of subrectangles arising from Theorem~\ref{duplicating_configurations} lead to unavoidable stopping sets in $\mathscr{L}^2_q$-TD LDPC codes. 
Since \mbox{$\ell+|\mathcal{S}^*_{\mathscr{C}_2}|$} is a lower bound for the size of these stopping sets, it is important to maximize the number of unique symbols $\mathcal{S}^*_{\mathscr{C}_2}$ for as many correlating subrectangles $(\mathscr{C}_1,\mathscr{C}_2)$ (of which $\mathscr{C}_1$ is full) as possible. We demonstrate this maximization with the 6-polygon illustrated by Fig.~\ref{upto7cellconfigs}~(b).
Fig.~\ref{fig:representation_6sided_polygon}~(a) represents all subrectangles that are isomorphic to this 6-polygon in $\mathcal{L}^{(\alpha_1)}_q$ and
(b) represents all correlating pairs in the orthogonal Latin square $\mathcal{L}^{(\alpha_2)}_q$. Hence, by the same choice of $x_1, x_2, x_3,y_1\in \mathbb{F}_q$ such that $x_1<x_2<x_3$, we obtain a correlating pair of subrectangles, denoted by $(\mathscr{C}_1, \mathscr{C}_2)$. Note that $\mathscr{C}_1$ in $\mathcal{L}^{(\alpha_1)}_q$ is full, while $\mathscr{C}_2$ in $\mathcal{L}^{(\alpha_2)}_q$ is generally not full.

Fig.~\ref{fig:representation_6sided_polygon}~(i)-(vi) show the possible combinations of entries of $\mathscr{C}_2$ that might be equal. Note that several combinations can not appear, since $\mathscr{C}_2$ must be orthogonal to $\mathscr{C}_1$. We now equate the respective entries of the variable representation of $\mathscr{C}_2$ given by Fig.~\ref{fig:representation_6sided_polygon}~(b), leading to the following equations (over $\mathbb{F}_q$):
\begin{eqnarray*}
\mbox{(i)} & \alpha_2 x_1 + (\alpha_1-\alpha_2) x_2 - \alpha_1 x_3 & = 0, \\
\mbox{(ii)} & (\alpha_2-\alpha_1) x_1 + \alpha_1 x_2 - \alpha_2 x_3 & = 0,\\
\mbox{(iii)} & -\alpha_2 x_1 + \alpha_1 x_2 + (\alpha_2 - \alpha_1) x_3 & = 0,\\
\mbox{(iv)} & -\alpha_1 x_1 + (\alpha_1-\alpha_2) x_2 + \alpha_2 x_3 & = 0,\\
\mbox{(v)} & (\alpha_2-\alpha_1) x_1 - \alpha_2 x_2 + \alpha_1 x_3 & = 0,\\
\mbox{(vi)} & -\alpha_1 x_1 + \alpha_2 x_2 + (\alpha_1-\alpha_2) x_3 & = 0.
\end{eqnarray*}

Every equation can be satisfied, for example, (i) by $x_1=\alpha_1$, $x_2=0$ and $x_3=\alpha_2$. 
Our approach is now to derive conditions that may prevent any two equations from being satisfied simultaneously. 
For example, by considering the linear system of the equations (i) and (ii) it follows that $(\alpha_1^2 -\alpha_1 \alpha_2 + \alpha_2^2)(x_3-x_2) = 0$ and thus $x_2 = x_3$ or $\alpha_1^2 -\alpha_1 \alpha_2 + \alpha_2^2 = 0$. The first condition can never be satisfied, since $x_2$ and $x_3$ represent two rows that must be distinct. The second condition depends only on the scale factors $\alpha_1$ and $\alpha_2$, and thus can be prevented from being satisfied by a proper choice of these scale factors. 
Similarly, we can derive such a condition for any pair of equations, leading to the following unique list of constraints that should be satisfied over $\mathbb{F}_q$:
\begin{align*}
2 \alpha_1 - \alpha_2 & \neq 0, \tag{C1}\\
2 \alpha_2 - \alpha_1 & \neq 0, \tag{C2}\\
\alpha_1 + \alpha_2 & \neq 0, \tag{C3}\\
\alpha_1^2 - \alpha_1 \alpha_2 + \alpha_2^2 & \neq 0. \tag{C4}
\end{align*}

By satisfying the constraints C1-C4 we can ensure that $|\mathcal{S}^*_{\mathscr{C}_2}|\geq 4$, since only two cells of $\mathscr{C}_2$ may have the same symbol. Hence, by applying Theorem~\ref{duplicating_configurations} to the correlating subrectangles $(\mathscr{C}_1,\mathscr{C}_2)$, we obtain full-correlating subrectangles $(\Psi_1,\Psi_2)$ of size $\ell'\geq \ell+|\mathcal{S}^*_{\mathscr{C}_2}|=10$.
We note that the first three constraints are the same as obtained in Subsection~\ref{avoiding_small_stopsets}.

\subsection{Simulation-based Optimization of $\mathscr{L}^2_q$-TD LDPC codes}

By exhaustive computer search, we collected a large number of size-10 stopping sets that occurred in various $\mathscr{L}^2_q$-TD LDPC codes. Based on these stopping sets, we revealed further frequently occurring constraints that should be satisfied:
\begin{align*}
\alpha_1^2+\alpha_1\alpha_2-\alpha_2^2 & \neq 0,\tag{C5}\\
\alpha_2^2+\alpha_1\alpha_2-\alpha_1^2 & \neq 0,\tag{C6}\\
\alpha_1^2-3\alpha_1\alpha_2+\alpha_2^2 & \neq 0.\tag{C7}
\end{align*}

\section{Quasi-Cyclic TD LDPC Codes}\label{encoding}

Let $\mathbb{F}_p$ be the Galois field of a prime order $p$ and let $(\mathcal{L}^{(\alpha_1,1)}_p,\hdots,\mathcal{L}^{(\alpha_m,1)}_p)\in\mathscr{L}_p^m$ be a set of MOLS in reduced form. We restrict that $\alpha_i\neq (p-1)$ for $1\leq i\leq m$. Now, we replace each $\mathcal{L}^{(\alpha_i,1)}_p$ by $\mathcal{L}^{(\alpha'_i,\beta'_i)}_p$ with $(\alpha'_i,\beta'_i)=\omega_i(\alpha_i,1)$ and  $\omega_i=(\alpha_i + 1)^{-1}$ over $\mathbb{F}_p$, which is possible according to Lemma~\ref{replacing_latin_squares} without changing the stopping set distribution of the corresponding $\mathscr{L}_p^m$-TD LDPC code.
For every cell-coordinates $(x,y) \in \mathbb{F}_p$ of the Latin squares, we obtain a column of the parity-check matrix $H$. This column has $(m+2)p$ positions, enumerated from $0$, with 1-entries at the positions \[\{x,\: p+y,\:  2p+\mathcal{L}^{(\alpha'_1,\beta'_1)}_p[x,y],\hdots,(m+1)p+\mathcal{L}^{(\alpha'_m,\beta'_m)}_p[x,y]\}\] and 0-entries at all remaining positions.
The order of the columns of $H$ corresponds to the order of processing the cells which is specified by the following sequence of cell-coordinates: \[[x,0],[x+1,1],\hdots,[x+(p-1),p-1]\pmod p\] for $x=0,\hdots,p-1$.
This approach leads to a parity-check matrix with quasi-cyclic structure, and hence, encoding can be realized via simple feedback shift registers \cite{TowWel67} achieving linear complexity in block length.

\emph{Remark: }
By replacing the MOLS as described above, we obtain $p$ diagonals within each Latin square, given by the sets of cell-coordinates $\{[x+i,i]: i=0,\hdots,p-1\} \pmod p$ for $x=0,\hdots,p-1$. These cells contain all symbol numbers from $0$ to $p-1$ in ascending order. For each diagonal, we process the $p$ cells consecutively such that every step increases the row index, the column index and the symbol numbers within the Latin squares by one (modulo $p$), leading to circulant submatrices in the parity-check matrix.

\section{Simulation Results}\label{simulations}

In Table~\ref{table_stopsets}, we simulated the transmission and decoding for $2\cdot 10^8$ codewords of all possible $\mathcal{L}_{13}^2$-TD LDPC codes over the BEC with error probability $\epsilon=0.075$ and counted the occurrences of stopping sets separately for every size up to 12. 
Detections of sizes greater than 12 have not been listed, since they are not important for our considerations. 
For recovering the erasures, we used the algorithm described in \cite{Luby1997}, which is often referred to as the \emph{erasure decoder} or \emph{peeling decoder}.
Each code given in Table \ref{table_stopsets} is specified by a pair $(\alpha_1,\alpha_2)$, meaning that the code is based on the Latin squares $(\mathcal{L}_{13}^{(\alpha_1)},\mathcal{L}_{13}^{(\alpha_2)})\in\mathscr{L}_{13}^2$. Since $q=13$, these codes are $(4,13)$-regular LDPC codes of block length 169 and rate 0.71. Additionally, we give the \emph{bit error rate} (BER) for every code and list the constraints C1-C7 that are violated by the code parameters $(\alpha_1,\alpha_2)$. 
As we can see, the constraints C1-C3 should have highest precedence, since their violation leads to many stopping sets of size $8$. Moreover, by the satisfaction of C1-C3 no stopping sets of size 8 occurred. Also, no stopping sets of size 9 occurred. Hence, we can raise the stopping distance from 8 to 10. Furthermore, the satisfaction of constraint C4 reduces the number of \mbox{size-10} stopping sets significantly. The constraints C5-C7 are not violated by these codes. 
Note that for an increasingly smaller error probability~$\epsilon$ stopping set detections of higher sizes are successively vanishing and thus stopping sets of small sizes dominate the decoding performance at the error-floor region. Hence, it is sufficient to reduce the occurrences of small stopping sets for lowering the error-floors.

\begin{table}[!t]
\renewcommand{\arraystretch}{1.3}
\caption{Stopping Set Detections for all $\mathcal{L}_{13}^2$-TD LDPC Codes Over the BEC Channel With Erasure Probability $\epsilon=0.075$}\label{table_stopsets}
\begin{center}
\scalebox{0.8}{
\begin{tabular}{|c||c|c|c|c|c|c||c||c|}
\hline
\multirow{2}{*}{$\boldsymbol{(\alpha_1,\alpha_2)}$} & \multicolumn{6}{c||}{\textbf{Stopping Set Detections}} & \multirow{2}{*}{\textbf{Bit Error Rate}} & \multirow{2}{*}{\textbf{Violations}}\\
\cline{2-7} & $\boldsymbol{\leq 7}$ & $\boldsymbol{8}$ & $\boldsymbol{9}$ & $\boldsymbol{10}$ & $\boldsymbol{11}$ & $\boldsymbol{12}$ & & \\
\hline\hline
$(1,2)$ & 0 & 517 & 0 & 25 & 3 & 15 & $16.30\cdot 10^{-8}$ & C1 \\\hline
$(1,3)$ & 0 & 0 & 0 & 13 & 1 & 6 & $3.90 \cdot 10^{-8}$ & \\\hline
$(1,4)$ & 0 & 0 & 0 & 34 & 1 & 6 & $4.39 \cdot 10^{-8}$ & C4 \\\hline
$(1,5)$ & 0 & 0 & 0 & 10 & 2 & 6 & $3.74 \cdot 10^{-8}$ & \\\hline
$(1,6)$ & 0 & 0 & 0 & 7 & 3 & 6 & $3.46 \cdot 10^{-8}$ & \\\hline
$(1,7)$ & 0 & 521 & 0 & 29 & 1 & 6 & $15.95 \cdot 10^{-8}$ & C2\\\hline
$(1,8)$ & 0 & 0 & 0 & 4 & 1 & 5 & $3.91 \cdot 10^{-8}$ & \\\hline
$(1,9)$ & 0 & 0 & 0 & 15 & 3 & 7 & $3.76 \cdot 10^{-8}$ & \\\hline
$(1,10)$ & 0 & 0 & 0 & 41 & 1 & 10 & $4.15 \cdot 10^{-8}$ & C4\\\hline
$(1,11)$ & 0 & 0 & 0 & 9 & 1 & 7 & $3.80 \cdot 10^{-8}$ & \\\hline
$(1,12)$ & 0 & 522 & 0 & 31 & 5 & 9 & $16.49 \cdot 10^{-8}$ & C3\\\hline
\end{tabular}
}
\end{center}
\end{table}

\begin{figure*}[t!]
		\centerline{
			\subfloat{
				\includegraphics[scale = 0.67]{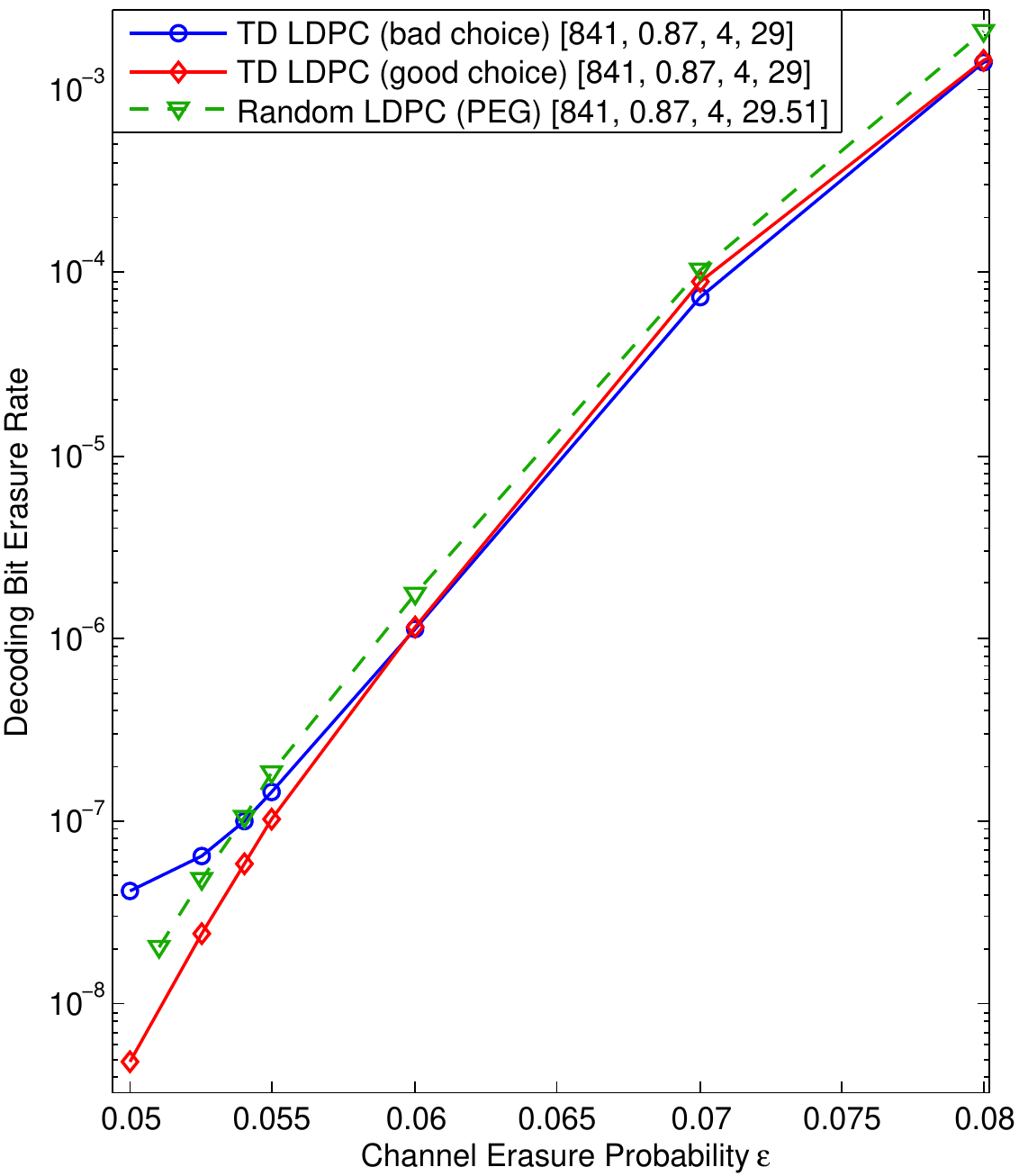}
			}
			\hspace{1cm}
			\vspace{0.3cm}
			\subfloat{
				\includegraphics[scale = 0.67]{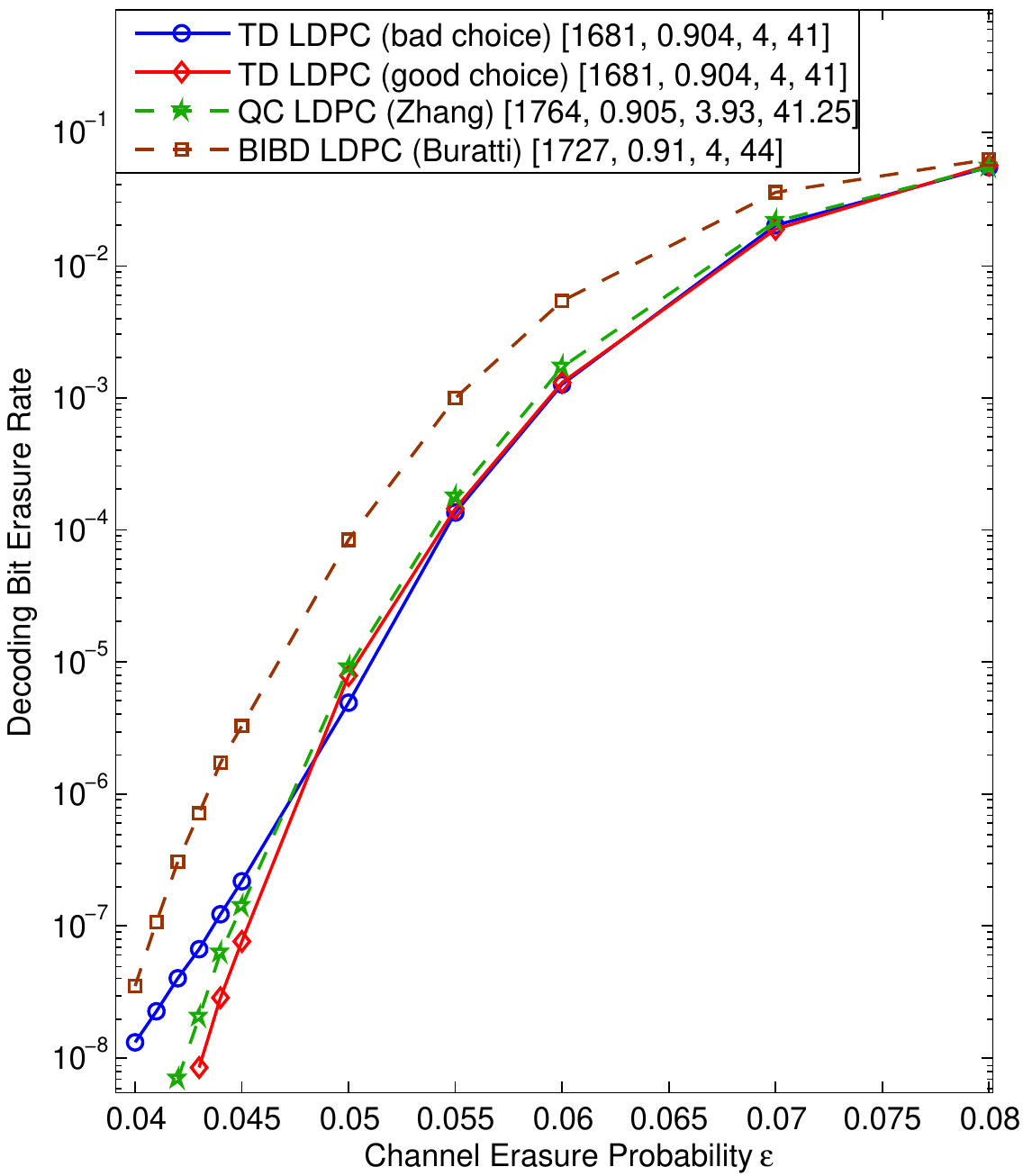}
			}
		}
		
		\centerline{
			\subfloat{
				\includegraphics[scale = 0.67]{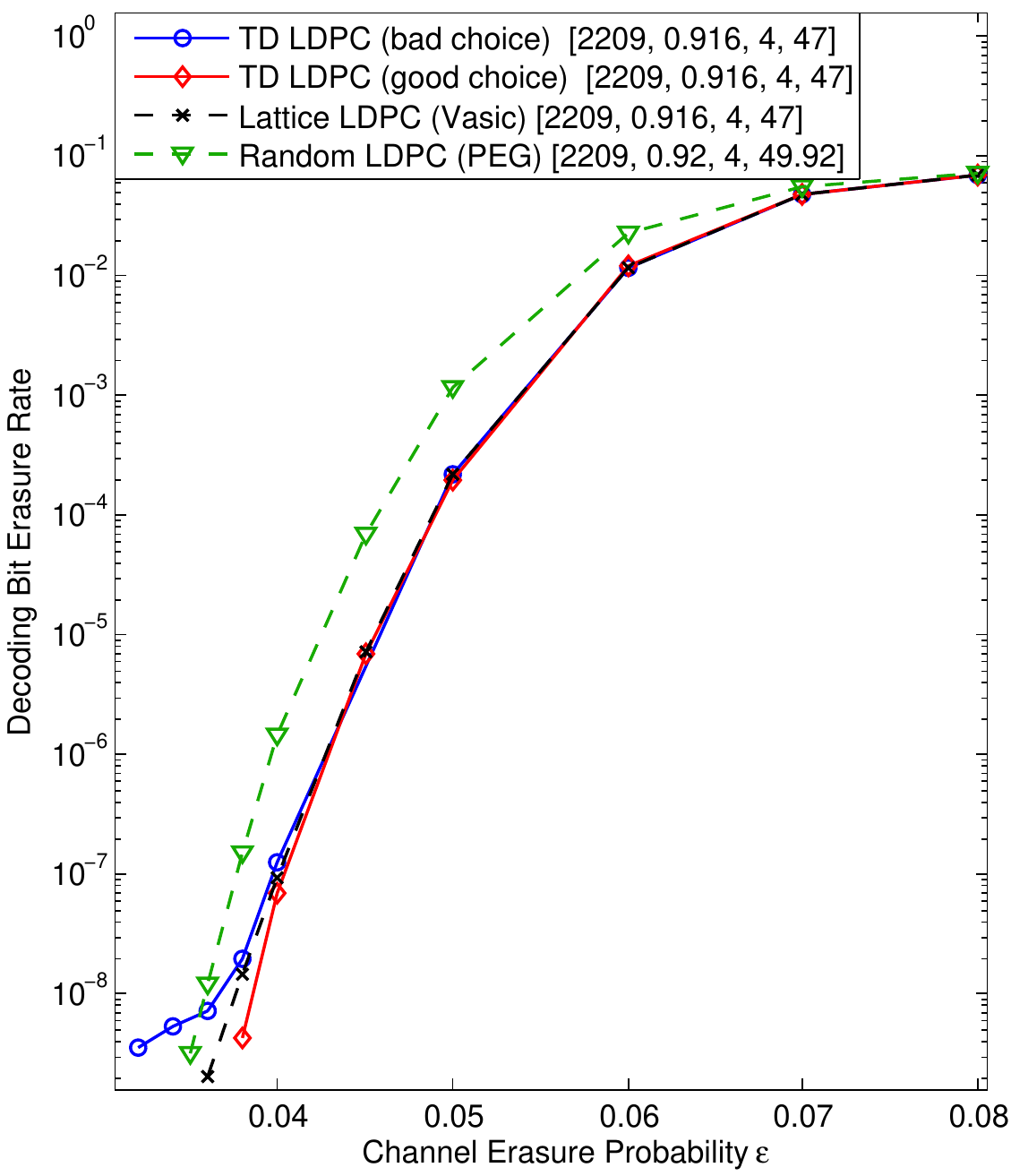}
			}
			\hspace{1cm}
			\vspace{0.3cm}
			\subfloat{
				\includegraphics[scale = 0.67]{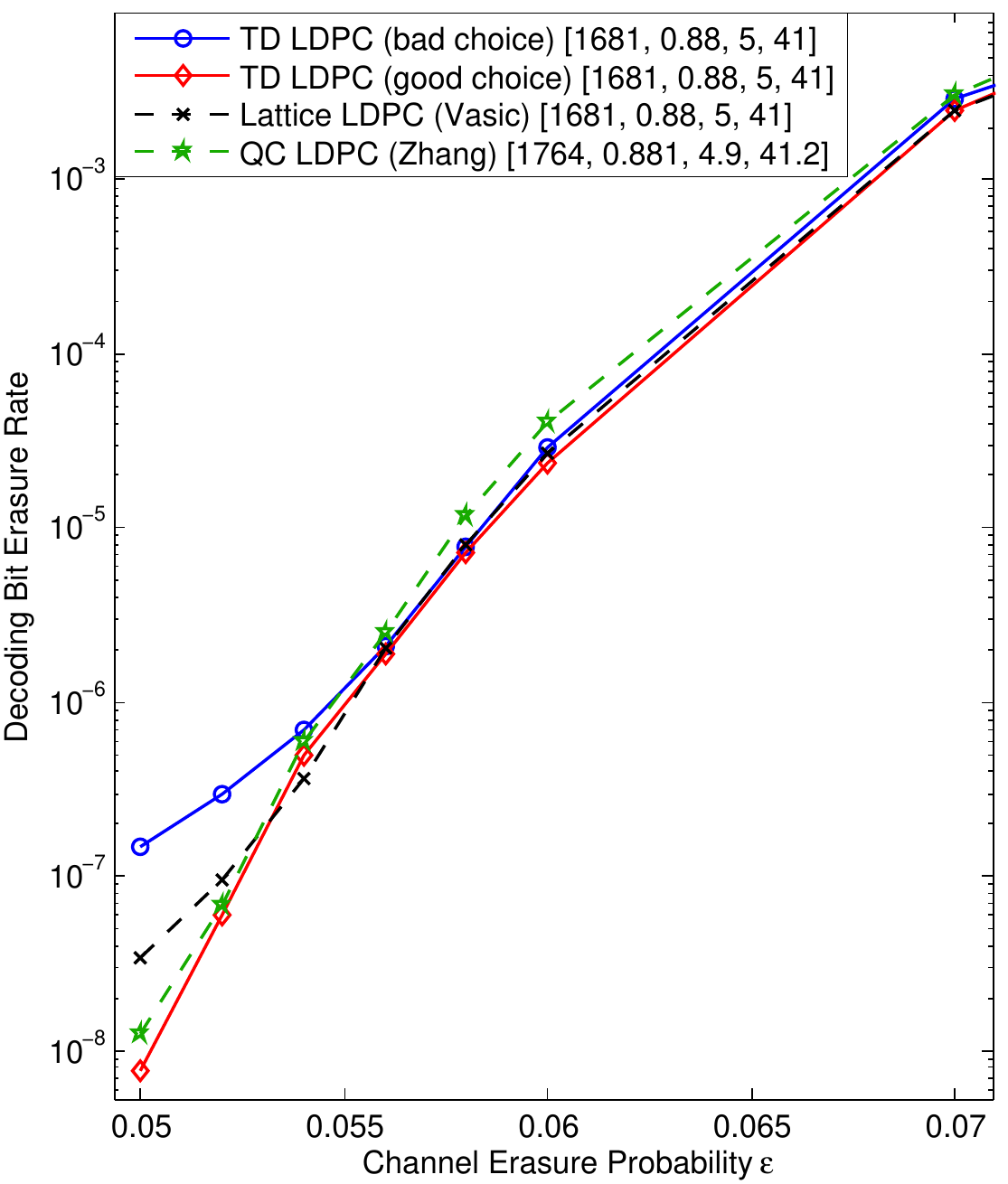}
			}
		}
	\caption{Decoding performance of various TD LDPC codes}
	\label{Simu}
\end{figure*}

For further experimental results on decoding performance (Fig.~\ref{Simu}), we again employed the peeling decoder over the BEC. A legend displays the following information in the respective order: code type, annotations in brackets, and a quadruple $[N, R, k, r]$ consisting of block length $N$, code rate $R$, column weight $k$ and row weight $r$. For irregular LDPC codes, the row and column weights are averaged. 
The first plot of Fig.~\ref{Simu} shows the decoding performance of two $\mathscr{L}_{29}^2$-TD LDPC codes based on the orthogonal Latin squares $(\mathcal{L}_{29}^{(1)},\mathcal{L}_{29}^{(2)})$ (bad choice) and $(\mathcal{L}_{29}^{(1)},\mathcal{L}_{29}^{(3)})$ (good choice), respectively, compared to a random regular LDPC code of nearly the same parameters constructed by the \emph{progressive edge-growth} (PEG) algorithm \cite{Hu2005}. The first $\mathscr{L}_{29}^2$-TD LDPC code suffers a high error-floor at approximately $\epsilon=0.06$, while the code with careful chosen scale factors shows a significantly better decoding performance at this region, even outperforming the randomly constructed code. 
The second plot illustrates the performance of $\mathscr{L}_{41}^2$-TD LDPC codes compared to a structural BIBD LDPC code based on the construction of Buratti (c.f.~\cite{Vasic04,GrunHub12}) and a quasi-cyclic (QC) LDPC code from \cite{Zhang2010}, which is also based on the concept of Latin squares. Again, by a suitable choice of the scale factors, we can improve the decoding performance in the error-floor region substantially, similar to the QC LDPC code by Zhang~et~al. \cite{Zhang2010}. Note that for a better comparison, we truncated the BIBD LDPC code in length without waiving the regularity of the code. 
The third plot shows that the $\mathscr{L}_{47}^2$-TD LDPC codes with proper chosen scale factors outperform a randomly PEG-constructed LDPC code of similar parameters and an LDPC code based on the Lattice construction from \cite{Vasic04} with exactly the same parameters. Although the Lattice codes are a subclass of the TD LDPC codes (cf. Subsection~\ref{related}), they seem to be disadvantaged due to their strict determination of the scale factors. Again, we observe an error-floor for the $\mathscr{L}_{47}^2$-TD LDPC code with the adverse choice of scale factors at an error probability of approximately $\epsilon=0.04$.
The last plot demonstrates that our results are also beneficial for the construction of $\mathscr{L}_{41}^3$-TD LDPC codes with column weight~5. The code based on the Latin squares $(\mathcal{L}_{41}^{(1)},\mathcal{L}_{41}^{(2)},\mathcal{L}_{41}^{(40)})$ is a bad choice with respect to the constraints C1-C7 and subsequently exhibits a high error-floor. The second code based on $(\mathcal{L}_{41}^{(1)},\mathcal{L}_{41}^{(3)},\mathcal{L}_{41}^{(9)})$ is a good choice, since its scale factors pairwise satisfy our constraints.

\section{Discussion and Related Work}\label{related}

LDPC codes based on transversal designs were first introduced in \cite{JohnWell2004} as a subclass of codes from \emph{partial geometries}. The code parameters of these LDPC codes are as in Subsection~\ref{codes_from_TDs}, but given in terms of partial geometries. The transversal designs arise from MOLS that are similarly constructed as in Lemma~\ref{simple_structured_MOLS}, but with only one scale factor $\alpha$ (and $\beta=1$). By contrast, the generalized construction of MOLS with two scale factors allows us to design quasi-cyclic TD LDPC codes with low-complexity encoding as described in Section~\ref{encoding}.
The minimum distance of any TD LDPC code is lower bounded in \cite{JohnWell2004} by $d_{min}\geq k+1$, where $k$ is the column weight of the code's parity-check matrix. In \cite{JohnsonDiss}, a method is presented  to generate transversal designs of block size $3$ without \emph{Pasch configurations}, leading to TD LDPC codes of column weight $3$ and minimum distance $d_{min} = 6$.
A general lower bound for the stopping distance of LDPC codes based on BIBD's of block size $k$ has been established in \cite{Ka03} by $s_{min}=k+1$. Note that transversal designs can be seen as partial BIBD's for which the same lower bound must be valid. 
The stopping distance of an $\mathscr{L}^1_q$-TD LDPC code of column weight 3 is $s_{min}=6$ for $\omega_q>2$ (Theorem~\ref{stopdist_weight3}) and the stopping distance of an $\mathscr{L}^2_q$-TD LDPC code of column weight $4$ is $s_{min}\geq 8$ for $\omega_q>3$ (Theorem~\ref{stopdist_weight4}). As we can see, both results exceed the general lower bound significantly. Furthermore, by optimizing the stopping set distribution of $\mathscr{L}^2_q$-TD LDPC codes, we can even raise the stopping distance from from 8 to 10.

Sets of $m$ MOLS can also be used to construct $q$-ary \emph{maximum distance separable (MDS)} codes \cite{Singleton64} ($q$ being the order of the constituent Latin squares) with block length $N=m+2$, rate $R=2/(m+2)$ and minimum distance $d=m+1$ \cite{ChangPark01}, achieving equality of the \emph{Singleton bound}. Although these codes arise from the same combinatorial entities as the TD LDPC codes considered in our paper, there are no obvious relations between both code classes. 

In~\cite{Vasic04}, special Latin squares are used to construct \emph{Steiner triple systems (STSs)}, which are certain subclasses of BIBDs, based on Bose's construction \cite{LindRod09}. The incidence matrix of any resulting STS can be considered as the parity-check matrix of a $(3,3q+1)$-regular LDPC code of block length $N=(6q+3)(6q+2)/6$ and rate $R\geq (3q-2)/3$, where $q$ is the order of an idempotent and commutative Latin square. If $2q+1\equiv 0 \mod 7$, the resulting STS is free of Pasch-configurations and thus the minimum distance of the LDPC code is $d_{min} = 6$. The results are extended in~\cite{LaenMil07} by a general investigation of the occurring substructures such as cycles, small stopping sets and small trapping sets, which have an influence on the decoding performances over the BEC and AWGN channels.
The existence of these substructures are linked to the existence of partial subrectangles in the underlying Latin squares, giving a simplified setting for  investigation. In the present paper, this idea is adapted for TD LDPC codes.

In \cite{Zhang2010}, quasi-cyclic LDPC codes based on Latin squares are constructed with a method called matrix dispersion. 
This technique can produce highly redundant parity-check matrices, leading to an excellent decoding performance at the cost of extra decoding computation. The paper provides an extensive analysis of the ranks of these parity-check matrices. However, error-floors has not been considered. 

A code class that is closely related to the presented TD LDPC codes is based on the geometric concept of rectangular integer lattices as proposed in \cite{Vasic04}. These codes can be seen as a subclass of TD LDPC codes. Given a prime $q$ and an integer $c$ with $3\leq c\leq q$, the $(q, c)$-parameterized code from the lattice construction (cf. \cite{Vasic04}) is equivalent to an TD LDPC code of column weight $c$ based on the set of MOLS $\{\mathcal{L}_q^{(\alpha_i, \beta_i)}: 1\leq i \leq c-2\}$ with scale factors $\alpha_i = q-i$ and $\beta_i = i+1$. 
After reducing these MOLS to the form $\mathcal{L}_q^{(\alpha'_i, 1)}$ with scale factors $\alpha'_i = (q-i)(i+1)^{-1}\bmod q$ according to Theorem \ref{reduced_form}, we note that the predetermined shift factors $\alpha'_i$ may not pairwise satisfy the constraints C1-C7.  For example, with $(q,c)=(5,4)$, we obtain a $\mathscr{L}_{5}^{2}$-TD LDPC code based on the MOLS $\{\mathcal{L}_4^{(4,2)},\mathcal{L}_4^{(3,3)}\}$ which can be reduced to $\{\mathcal{L}_4^{(2,1)},\mathcal{L}_4^{(1,1)}\}$. Here, the scale factors $\alpha_1=2$ and $\alpha_2=1$ violate constraint C2, leading to harmful stopping sets in the arising code.

\section{Conclusion}\label{conclusion}

In this paper, we have investigated stopping sets of LDPC codes from transversal designs based on mutually orthogonal Latin squares. We established a link between full configurations in transversal designs (corresponding to stopping sets in TD LDPC codes) and special substructures in the underlying Latin squares, which we have introduced as full-correlating subrectangles. This new concept allows a thorough examination of stopping sets in the resulting codes. 
The Latin squares considered in this paper possess a simple structure and depend only on a pair of constants, called the scale factors. Based on this class of Latin squares, we bounded the stopping distance of TD LDPC codes of column weight 3 and~4.
Furthermore, we derived constraints for the choice of well-matching scale factors, leading to improved stopping set distributions for the arising codes. These constraints are derived for column weight~4, but are also potentially beneficial for TD LDPC codes with higher column weights. 
As a consequence, our approach leads to codes with significantly lower error-floors over the binary erasure channel. We demonstrated the performance gain of our novel code design by extensive simulations for TD LDPC codes with small to moderate block lengths.

\section*{Acknowledgment}

The authors thank the anonymous referees for their careful reading and valuable insights that helped improving the presentation of the paper.

\newpage

\begin{biography}[{\includegraphics[width=1in, height=1.25in,clip,keepaspectratio]{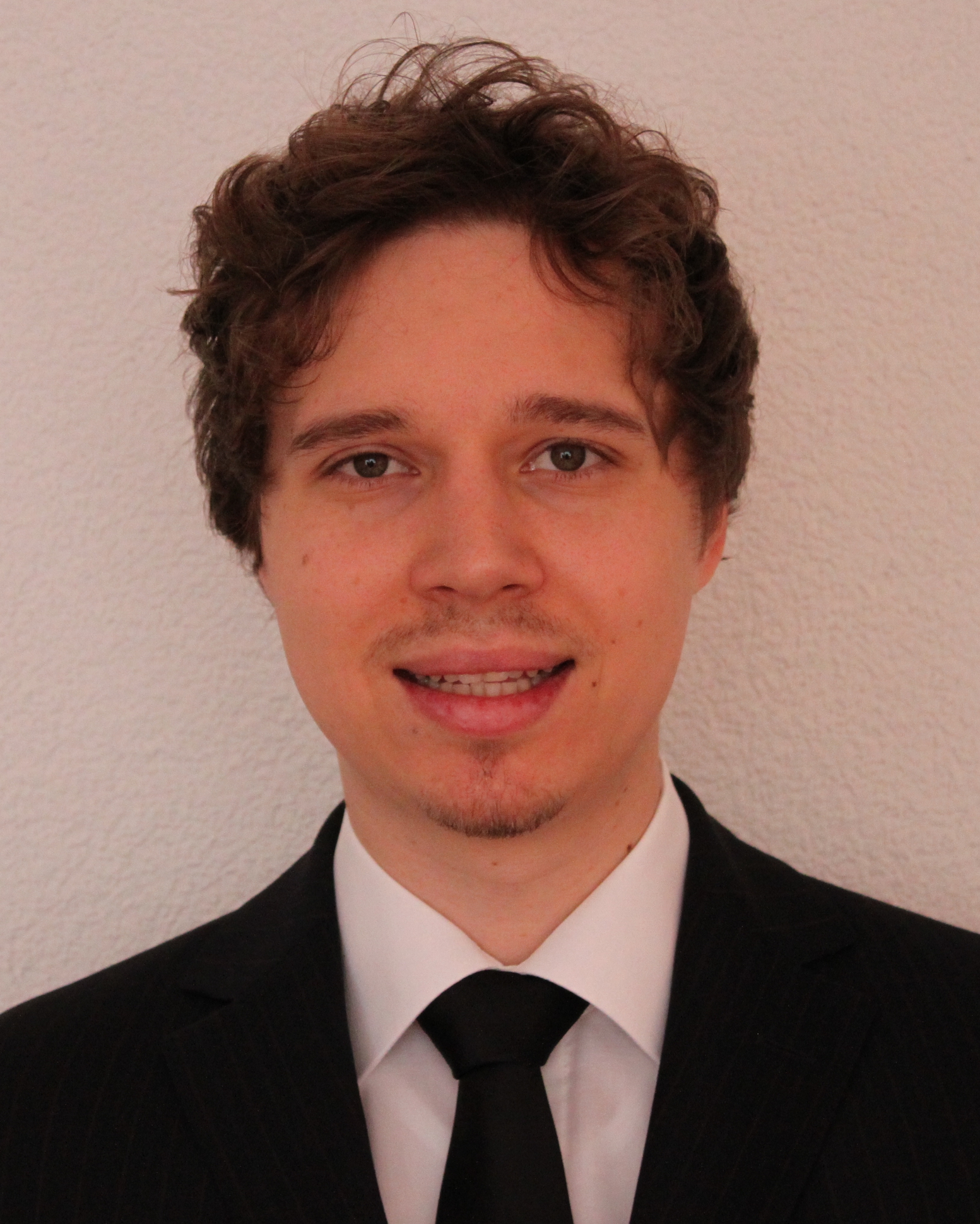}}]{Alexander Gruner} is a Ph.D. student in computer science at the Wilhelm Schickard Institute for Computer Science, University of T{\"u}bingen, Germany, where he is part of an interdisciplinary research training group in computer science and mathematics. He received the Diploma degree in computer science from the University of T{\"u}bingen in 2011. His research interests are in the field of coding and information theory with special emphasis on turbo-like codes, codes on graphs and iterative decoding.
\end{biography}

\begin{biography}[{\includegraphics[width=1in, height=1.25in,clip,keepaspectratio]{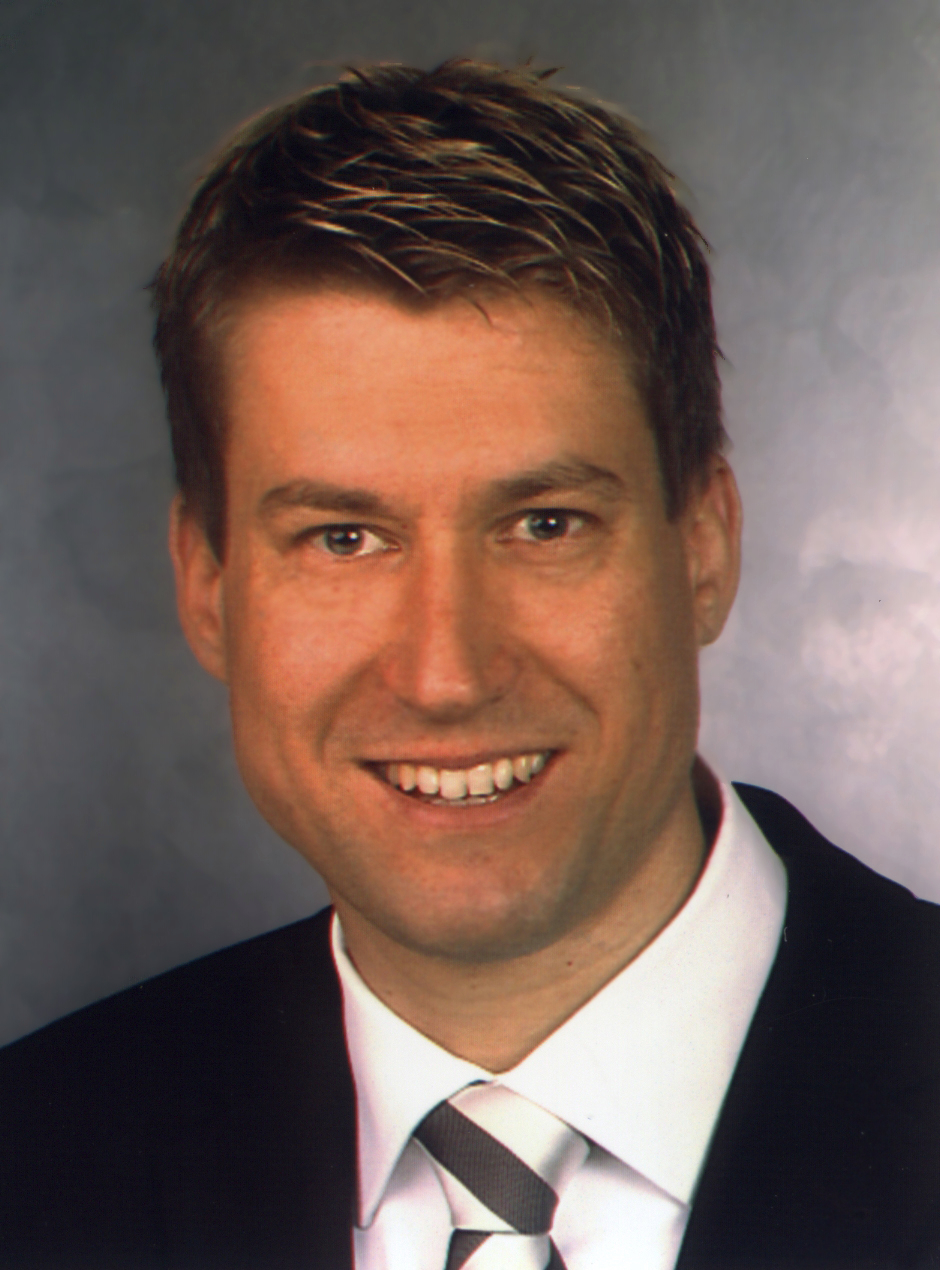}}]{Michael Huber} (M.'09) is currently a manager for data quality and business intelligence at Mercedes-Benz Bank AG, Stuttgart, and adjunct professor at the Wilhelm Schickard Institute for Computer Science, University of T\"ubingen, Germany. He was a Heisenberg fellow of the German Research Foundation (DFG) from 2008-12 and a visiting full professor in mathematics at Berlin Technical University from 2007-08. He obtained the Diploma, Ph.D. and Habilitation degrees in mathematics from the University of T\"ubingen in 1999, 2001 and 2006, respectively. He was awarded the 2008 Heinz Maier Leibnitz Prize by the DFG and the German Ministry of Education and Research (BMBF). He became a Fellow of the Institute of Combinatorics and Its Applications (ICA), Winnipeg, Canada, in 2009.

Prof. Huber's research interests are in the areas of coding and information theory, cryptography and information security, discrete mathematics and combinatorics, algorithms, and data analysis and visualization with applications to business and biological sciences. Among his publications in these areas are two books, Flag-transitive Steiner Designs (Birkh\"auser Verlag, Frontiers in Mathematics, 2009) and Combinatorial Designs for Authentication and Secrecy Codes (NOW Publishers, Foundations and Trends in Communications and Information Theory, 2010). He is a Co-Investigator of an interdisciplinary research training group in computer science and mathematics at the University of T\"ubingen.
\end{biography}

\enlargethispage{-4.3in}

\end{document}